\documentclass[sigconf]{acmart}

\renewcommand\footnotetextcopyrightpermission[1]{}
\usepackage{balance}

\def\BibTeX{{\rm B\kern-.05em{\sc i\kern-.025em b}\kern-.08emT\kern-.1667em\lower.7ex\hbox{E}\kern-.125emX}}

\copyrightyear{2019}
\acmYear{2019}
\setcopyright{none}
\acmConference[SPAA '19]{31st ACM Symposium on Parallelism in Algorithms and Architectures}{June 22--24, 2019}{Phoenix, AZ, USA}
\acmBooktitle{31st ACM Symposium on Parallelism in Algorithms and Architectures (SPAA '19), June 22--24, 2019, Phoenix, AZ, USA}
\acmDOI{10.1145/3323165.3323202}
\acmISBN{978-1-4503-6184-2/19/06}

\usepackage{graphicx}
\usepackage{amsmath}
\usepackage{amssymb}
\usepackage{algorithm2e}
\usepackage{nicefrac}
\usepackage{amsfonts,amsthm}
\usepackage[T1]{fontenc}

\newtheorem{theorem}{Theorem}[section]
\newtheorem{lemma}[theorem]{Lemma}

\newtheorem{invariant}[theorem]{Invariant}

\makeatletter
\newcommand{\fixed@sra}{$\vrule height 2\fontdimen22\textfont2 width 0pt\shortrightarrow$}
\newcommand{\shortarrow}[1]{%
  \mathrel{\text{\rotatebox[origin=c]{\numexpr#1*45}{\fixed@sra}}}
}
\makeatother

\newcommand{\sidenote}[1]{\textbf{(*)}\marginpar {\tiny \raggedright{(*) #1}}}

\newcommand{\neighborsbounds}{\sqrt{2m}}
\newcommand{\Q}{\ensuremath{\mathcal{Q}}}

\newcommand{\lvl}[1]{\ensuremath{lvl(#1)}}
\newcommand{\outset}[1]{\ensuremath{Out_{#1}}}
\newcommand{\inset}[2]{\ensuremath{In_{#1}[#2]}}
\newcommand{\lowneighbors}[2]{\ensuremath{\Phi_{#1}(#2)}}

\newcommand{\freeschedule}[0]{\textit{free-schedule}}
\newcommand{\unmatchschedule}[0]{\textit{unmatch-schedule}}
\newcommand{\riseschedule}[0]{\textit{rise-schedule}}
\newcommand{\shuffleschedule}[0]{\textit{shuffle-schedule}}
\newcommand{\handlefree}[0]{\textit{handle-free}}
\newcommand{\setlevel}[0]{\textit{set-level}}

\newcommand{\getalive}{getAlive}

\newcommand{\getdegreeinmachine}{getDegInMachine}
\newcommand{\getsuspended}{getSuspended}
\newcommand{\movesuspended}{moveSuspended}
\newcommand{\fits}{fits}
\newcommand{\tofit}{toFit}
\newcommand{\updatevertex}{updateVertex}
\newcommand{\moveedges}{moveEdges}
\newcommand{\fetchmore}{fetchSuspended}

\newcommand{\updateMachine}{updateMachine}

\newcommand{\addedge}{addEdge}

\newcommand{\ignore}[1]{}

\begin{document}

\title{Dynamic Algorithms for the Massively Parallel Computation Model}
\author[G. F. Italiano]{Giuseppe F. Italiano }
\affiliation{LUISS University, Rome, Italy}
\author[S. Lattanzi]{Silvio Lattanzi }
\affiliation{Google Research, Z\"urich, Switzerland}
\author[V. S. Mirrokni]{Vahab S. Mirrokni }
\affiliation{Google Research, New York, USA}

\author[N. Parotsidis]{Nikos Parotsidis}
\affiliation{University of Copenhagen, Denmark}
\authornote{The author is supported by Grant Number 16582, Basic Algorithms Research Copenhagen (BARC), from the VILLUM Foundation. Work partially done while the author was an intern at Google.}

\begin{abstract}
The Massive Parallel Computing (MPC) model gained popularity during the last decade and it is now seen as the standard model for processing large scale data. One significant shortcoming of the model is that it assumes to work on static datasets while, in practice, real world datasets evolve continuously. To overcome this issue, in this paper we initiate the study of dynamic algorithms in the MPC model.
We first discuss the main requirements for a dynamic parallel model and we show how to adapt the classic MPC model to capture them. Then we analyze the connection between classic dynamic algorithms and dynamic algorithms in the MPC model. Finally, we provide new efficient dynamic MPC algorithms for a variety of fundamental graph problems, including connectivity, minimum spanning tree and matching.
\end{abstract}


\maketitle

\pagenumbering{arabic} 

\section{Introduction}

Modern applications often require performing computations on massive amounts of data. Traditional models of computation, such as the RAM model or even shared-memory parallel systems, are inadequate for such computations, as the input data do not fit into the available memory of today's systems. The restrictions imposed by the limited memory in the available architectures has led to new models of computation that are more suitable for processing massive amounts of data. A model that captures the modern needs of computation at a massive scale is the Massive Parallel Computing (MPC) model, that is captured by several known systems (such as MapReduce, Hadoop, or Spark). At a very high-level, a MPC system consists of a collection of machines that can communicate with each other through indirect communication channels. The computation proceeds in synchronous rounds, where at each round the machines receive messages from other machines, perform local computations, and finally send appropriate messages to other machines so that the next round can start. The crucial factors in the analysis of algorithms in the MPC model are the number of rounds and the amount of communication performed per round.

The MPC model is an abstraction of a widely-used framework in practice and has resulted in an increased interest by the scientific community. An additional factor that contributed to the interest in this model is that MPC exhibits unique characteristics that are not seen in different parallel and distributed architectures, such as its ability to perform expensive local computation in each machine at each round of the computation. Despite its resemblance to other parallel models, such as the PRAM model, the MPC model has been shown to have different algorithmic power from the PRAM model~\cite{karloff2010model}.

The ability of the MPC model to process large amounts of data, however, comes with the cost of the use of large volumes of resources (processing time, memory, communication links) during the course of the computation. This need of resources strengthens the importance of efficient algorithms. Although the design of efficient algorithms for solving problems in the MPC model is of vital importance, applications often mandate the recomputation of the solution (to a given problem) after small modifications to the structure of the data. For instance, such applications include the dynamic structure of the Web where new pages appear or get deleted and new links get formed or removed, the evolving nature of social networks, road networks that undergo development and constructions, etc. In such scenarios, even the execution of very efficient algorithms after few modifications in the input data might be prohibitive due to their large processing time and resource requirements. Moreover, in many scenarios, small modifications in the input data often have a very small impact in the solution, compared to the solution in the input instance prior to the modifications. These considerations have been the driving force in the study of dynamic algorithms in the traditional sequential model of computation. 

Dynamic algorithms maintain a solution to a given problem throughout a sequence of modifications to the input data, such as insertions or deletion of a single element in the maintained dataset. In particular, dynamic algorithms are able to adjust efficiently the maintained solution by typically performing very limited computation. Moreover, they often detect almost instantly that the maintained solution needs no modification to remain a valid solution to the updated input data. The update time of a dynamic algorithm in the sequential model is the time required to update the solution so that it is a valid solution to the current state of the input data. Dynamic algorithms have worst-case update time $u(N)$ if they spend at most $O(u(N))$ after every update, and $u(N)$ amortized update bound if they spend a total of $O(k\cdot u(N))$ time to process a sequence of $k$ updates. The extensive study of dynamic algorithms has led to results that achieve a polynomial, and often exponential, speed-up compared to the recomputation of a solution from scratch using static algorithms, for a great variety of problems. For instance, computing the connected components of a graph takes $O(m+n)$ time, where $n$ and $m$ are the number of vertices and edges of the graph, respectively, while the most efficient dynamic algorithms can update the connected components after an edge insertion or an edge deletion in $O(\log n)$ amortized time~\cite{holm2001poly}, or in sub-polynomial time in the worst-case~\cite{lulli2017fast}. Similarly, there exist algorithms that can maintain a maximal matching in polylogarithmic time per update in the worst case~\cite{bernstein2019adeamortization}, while recomputing from scratch requires $O(m+n)$ time. 

So far, there has been very little progress on modelling dynamic parallel algorithms in modern distributed systems, despite their potential impact in modern applications, with respect to the speed-up and reduced use of resources. There have been few dynamic algorithms that maintain the solution to a problem in the distributed setting. For instance, in~\cite{censorhillel2016optimal}, Censor-Hillel et al. present a dynamic algorithm for maintaining a Maximal Independent Set of a graph in the LOCAL model. Assadi et al.~\cite{Assadi2018fully} improve the message complexity by adjusting their sequential dynamic algorithm to the LOCAL model. In~\cite{ahn2018access}, Ahn and Guha study problems that can be fixed locally (i.e., within a small neighborhood of some vertex) after some small modification that has a very limited impact on the existing solution. This line of work has been primarily concerned with minimizing the number of rounds and the communication complexity. Moreover, the algorithms designed for the LOCAL model do not necessarily take into account the restricted memory size in each machine. 

In this paper, we present an adaptation of the MPC model, that we call DMPC, that serves as a basis for dynamic algorithms in the MPC model. First, we impose a strict restriction on the availability of memory per machine, which mandates the algorithms in this model to operate in any system that can store the input in the total memory. Second, we define a set of factors that determine the complexity of a DMPC algorithm. These factors consist of (i) the number of rounds per update that are executed by the algorithm, (ii) the number of machines that are active per round, and (iii) the total amount of communication per round, which refers to the sum of sizes of all messages sent at any round. 
A final requirement for our model is that DMPC algorithms should provide worst-case update time guarantees. This is crucial not only because of the shared nature of the resources, but also because it  is imposed by many real-world applications, in which one needs to act fast upon an update in the data, such as detecting a new malicious behavior, or finding relevant information to display to a new activity (e.g., displaying ads, friend recommendations, or products that are relevant to a purchase).

Inspired by today's systems that share their resources between many different applications at any point in time, it is necessary to design algorithms that do not require dedicated systems to operate on, and that can be executed with limited amounts of resources, such as memory, processors, and communication channels. This necessity is further strengthened by the fact that  typically dynamic algorithms are required to maintain a solution to a problem over long series of updates, which implies that the application is running for a long sequence of time. Our model imposes these properties through the predefined set of restriction. In particular, we focus on three main dimensions

\paragraph{Memory.}
Dynamic algorithms in our model are required to use a very limited amount of memory in each machine. Specifically, assuming that the input is of size $N$, each machine is allowed to use only $O(\sqrt{N})$ memory. Note that this limitation does not aim at ensuring that the machines are large enough to fit $O(\sqrt{N})$ bits (as a system with such weak machines would need many millions of machines to even store the data, given that even weak physical machines have several GB of memory). Rather, it aims at guaranteeing that the allocation of the machines of the model to physical machines is flexible in terms of memory, allowing the system to move machines of the model across different physical machines without affecting the execution of the algorithm. (Notice that the system can co-locate several machines of the model to a single physical machine.)

\paragraph{Resource utilization and number of machines.} Our model promotes limited processing time in several ways. First, two factors of evaluation of an algorithm are the number of rounds that are required to process each update, and the number of machines that are active at each round of the update. Notice that machines that are not used by the execution of a dynamic algorithm can process other applications that co-exist in the same physical machines. 
Moreover, algorithms with worst-case update time are guaranteed to end the execution of a particular update in limited time, thus avoiding locking shared resources for large periods of time. 

\paragraph{Communication Channels.} In our model, one of the factors that contributes to the complexity of an algorithm is the amount of communication that occurs at each round during every update. 
Furthermore, the number of machines that are active per round also contributes to the complexity of an algorithm (namely, the number of machines receiving or transmitting messages).
These two facts ensure that efficient algorithms in the DMPC model use limited communication. 

\smallskip

Similarly to the sequential model, the goal of a dynamic 
algorithm in the DMPC model is to maintain a solution to a problem more efficiently than recomputing the solution from scratch with a static algorithm. Here, the main goal is to reduce the bounds in all three factors contributing to the complexity of an algorithm. However, algorithms reducing some of the factors, without increasing the others, may also be of interest. 

We initiate the study of dynamic algorithms in the DMPC model by designing algorithms for basic graph-theoretic problems. In particular, we present fully-dynamic algorithms for maintaining a maximal matching, a $\nicefrac{3}{2}$-approximate matching, a $(2+\epsilon)$-approximate matching, and the connected components of an unweighted graph, as well as a $(1+\epsilon)$-approximate Minimum Spanning Tree (MST) of a weighted graph.

Finally, we show that our model can exploit successfully the  techniques that were developed for dynamic algorithms in the sequential model.  In particular, we present a black-box reduction that transforms any sequential dynamic algorithm with $p(S)$ preprocessing time and $u(S)$ update time to an algorithm in the dynamic MPC model which performs the preprocessing step in $O(p(S))$ rounds, uses $O(1)$ machines and $O(1)$ total communication per round, and such that each update is performed in $O(u(S))$ number or rounds using $O(1)$ machines and $O(1)$ total communication per round. With this reduction, the characteristics (amortized vs. worst-time and randomized vs. deterministic) of the  DMPC algorithm are the same as the  sequential algorithm.

\smallskip\noindent\textbf{Related work in the classic MPC model.}
It was known from the PRAM model how to compute a $(1+\epsilon)$ approximate matching in $O(\log n)$ rounds~\cite{lotker2008improved}.
Lattanzi et al. ~\cite{lattanzi2011filtering} introduced the so-called filtering technique which gives an algorithm for computing a 2-approximate matching in $O(1/c)$ rounds assuming that the memory per machine is $O(n^{1+c})$, for any $c>0$.
Under the same memory assumption, Ahn and Guha \cite{ahn2018access} showed an algorithm running in $O(1/(c \epsilon))$ number of rounds for $(1+\epsilon)$ approximate matching.
Both those algorithms run in $O(\log n)$ time when the memory in each machine is $\Theta(n)$, which matches the bound that was known from the PRAM model.
It was only recently that Czumaj et al. \cite{Czumaj2018round} overcame the $O(\log n)$ barrier for computing an approximate matching.
In particular, in \cite{Czumaj2018round} the authors presented a $(1+\epsilon)$-approximate matching in $O((\log \log n)^2)$ time with $\widetilde{O}(n)$ memory per machine.
This bound has been improved to $O(\log \log n)$ rounds, under the assumption of slightly superlinear memory per machine ~\cite{ghaffari2018improved,assadi2019coresets}. Very recently, Ghaffari and Uitto~\cite{ghaffari2019sparsifying} presented an algorithm that uses only sublinear memory and can compute a $(1+\epsilon)$-approximate matching in $\widetilde{O}(\sqrt{\log \Delta})$ rounds, where $\Delta$ is the maximum degree in the graph. 

Another central problem in the MPC model is the computation of the connected components of a graph.
This problem can be solved in $O(\log n)$ rounds \cite{lulli2017fast,lkacki2018connected}.
In particular, the algorithm in \cite{lkacki2018connected} runs in $O(\log \log n)$ rounds on certain types of random graphs.
In the case where each machine contains $O(n^{1+c})$ memory, it is known how to compute the connected components of a graph in a constant number of rounds \cite{lattanzi2011filtering}.
Under a well-known conjecture \cite{yaroslavtsev2017massively}, it is impossible to achieve $o(\log n)$ on general graphs if the space per machine is $O(n^{1-c})$ and the total space in all machines is $O(m)$. Very recently Andoni et al.~\cite{DBLP:conf/focs/AndoniSSWZ18} presented a new algorithm that uses sublinear memory and runs in $\tilde{O}(\log D)$ parallel rounds, where $D$ is the diameter of the input graph.

\smallskip\noindent\textbf{Our results.}
Throughout the paper we denote by $G=(V,E)$ the input graph, and we use $n=|V|$, $m=|E|$, and $N=m+n$. All bounds that are presented in this section are worst-case update bounds. Our algorithmic results are summarized in Table \ref{tab:algorithms}. 
All of our algorithms use $O(N)$ memory across all machines, and hence make use of $O(\sqrt{N})$ machines.

\begin{table*}[t!]
	\centering
	\caption{Algorithmic results achieved in this paper. The bounds presented in the first part of the table hold in the worst-case.}
	\label{tab:algorithms}
	\renewcommand{\arraystretch}{1.2}
	\begin{tabular}{lcccc}
		\multicolumn{1}{c|}{Problem}            & \#rounds       & \begin{tabular}[c]{@{}c@{}}\#active \\ machines\end{tabular} & \multicolumn{1}{c|}{\begin{tabular}[c]{@{}c@{}}Commun.\\ per round\end{tabular}} & Comments                                                                                              \\ \hline
		\multicolumn{1}{l|}{Maximal matching}   & $O(1)$         & $O(1)$                                                       & \multicolumn{1}{c|}{$O(\sqrt{N})$}      & \begin{tabular}[c]{@{}c@{}}Use of a coordinator, \\ starts from an arbitrary graph. 
		\end{tabular}                                 \\ \hline
		\multicolumn{1}{l|}{3/2-app. matching}  & $O(1)$         & $O(n/\sqrt{N})$                                                & \multicolumn{1}{c|}{$O(\sqrt{N})$}      & \begin{tabular}[c]{@{}c@{}}Use of a coordinator.
		\end{tabular}                                                                                  \\ \hline

	\multicolumn{1}{l|}{$(2+\epsilon)$-app. matching}  & $O(1)$         & $\widetilde{O}(1)$  & 
		\multicolumn{1}{c|}{$\widetilde{O}(1)$} & \begin{tabular}[c]{@{}c@{}} 
		\end{tabular}                                                                                     \\ \hline
		\multicolumn{1}{l|}{Connected comps}    & $O(1)$         & $O(\sqrt{N})$                                                & \multicolumn{1}{c|}{$O(\sqrt{N})$}      & \begin{tabular}[c]{@{}c@{}} Use of Euler tours,\\starts from an arbitrary graph.
		\end{tabular}                                                                                                      \\ \hline
		\multicolumn{1}{l|}{$(1+\epsilon)$-MST} & $O(1)$         & $O(\sqrt{N})$                                                & \multicolumn{1}{c|}{$O(\sqrt{N})$}      & \begin{tabular}[c]{@{}c@{}} The approx. factor comes\\ from the preprocessing, \\ starts from an arbitrary graph.\end{tabular}                 \\ \hline
		\multicolumn{5}{c}{Results from reduction to the centralized dynamic model}                                                                                                                                                                                                                                                     \\ \hline
		\multicolumn{1}{c|}{Maximal matching}   & $O(1)$        & $O(1)$                                                       & \multicolumn{1}{c|}{$O(1)$}             & \begin{tabular}[c]{@{}c@{}}Amortized, randomized.
		\end{tabular}    \\ \hline
		\multicolumn{1}{c|}{Connected comps}    & $\widetilde{O}(1)$ & $O(1)$                                                       & \multicolumn{1}{c|}{$O(1)$}             & \begin{tabular}[c]{@{}c@{}}Amortized, deterministic.
		\end{tabular} \\ \hline
		\multicolumn{1}{c|}{MST}                & \ $\widetilde{O}(1)$ & $O(1)$                                                       & \multicolumn{1}{c|}{$O(1)$}          & \begin{tabular}[c]{@{}c@{}}Amortized, deterministic.
		\end{tabular}
	\end{tabular}
\end{table*}

\paragraph{Maximal matching.}
Our first algorithm maintains fully-dynamically a maximal matching in $O(1)$ rounds per update in the worst case, while the number of machines that are active per rounds is $O(1)$, and the total communication per round is $O(\sqrt{N})$. 
The general idea in this algorithm,  inspired from~\cite{neiman2016simple}, is to use vertex-partitioning across the machines and additionally to store at one machine the last $\sqrt{N}$ updates in a buffer, together with the changes that each of these updates generated. 
We call this summary of updates and the changes that they trigger the \emph{update-history}.
Every time that an update arrives (i.e., an edge insertion or an edge deletion), the update-history is sent to the endpoints that are involved in the update, and each endpoint adjusts its data structure based on the update-history (that is, it updates its knowledge of which vertices among its neighbors are free), and further sends back (to the machine that maintains the update-history) any possible changes that the update might have triggered. 
The machines that maintain the endpoints of the updated edge might further communicate with one of their neighbors to get matched with them.
Additional challenges arise from the fact that the neighborhood of a single vertex might not fit in a single machine.

For comparison, the best static MPC algorithm to compute a maximal matching in the static case runs in $O(\log \log n)$ when the space per machine is  $\widetilde{O}(n)$ \cite{ghaffari2018improved}, $O(\sqrt{ \log n})$ when the space is sublinear \cite{ghaffari2019sparsifying} and in $O(c/\delta)$ rounds when $N\in \Omega(n^{1+c})$ and the space per machine is  $\Omega(n^{1+\delta})$ \cite{lattanzi2011filtering}. 
These algorithms use all the machines at each round and generate $\Omega(N)$ communication per round.

We note that although our algorithm has communication complexity $O(\sqrt{N})$ per round in the case where the available memory per machine is $O(\sqrt{N})$, the communication complexity is actually proportional to the number of machines used by the system. Namely, if we allow larger memory per machine then the communication complexity reduces significantly. Hence, in real-world systems we expect our algorithm to use limited communication per MPC round.

\paragraph{3/2-approximate matching.}
We further study the problem of maintaining a maximum cardinality matching beyond the factor 2 approximation given by a maximal matching.
We present an algorithm for maintaining a $3/2$-approximate matching that runs in a constant number of rounds, uses $O(\sqrt{N})$ machines per round and with $O(\sqrt{N})$ communication per round.
The best known static algorithm for computing a $O(1+\epsilon)$ approximate matching runs in $O(\log \log n)$ rounds in the case where the memory available in each machine is $ \widetilde{O}(n)$~\cite{Czumaj2018round,ghaffari2018improved,assadi2019coresets} or in $O(\sqrt{ \log \Delta})$ rounds when the memory available in each machine is sublinear~\cite{yaroslavtsev2017massively}, where $\Delta$ the maximum degree in the graph.

\paragraph{$(2+\epsilon)$-approximate matching.}
Our algorithm for maintaining a maximal matching requires polynomial communication among the machines and the use of a coordinator machine. To overcome those restrictions, we explore the setting where we are allowed to maintain an almost maximal matching instead of a proper maximal matching. In other terms, 
at most an $\epsilon$ fraction of the edges of a maximal matching may be missing. 
In this setting, we show that we can adapt the fully-dynamic centralized algorithm by Charikar and Solomon \cite{charikar2018fully} that has polylogarithmic worst-case update time.
We note that our black-box reduction to the DMPC model yields a fully-dynamic algorithm with a polylogarithmic number of rounds.
However we show how we can adapt the algorithm to run in $O(1)$ rounds per edge insertion or deletion, using $O(\textit{polylog}(n))$ number of active machines and total communication per round. 
\footnote{We note that one could adopt the algorithm from \cite{bernstein2019adeamortization} to maintain a (proper) maximal matching with the same asymptotic bounds; however, that algorithm does not maintain a consistent matching throughout its execution,
meaning that the maintained matching could be completely different between consecutive update operations, which is not a desirable property for many applications.}

\paragraph{Connected components and $(1+\epsilon)$ MST}
We consider the problem of maintaining the connected components of a graph and the problem of maintaining a $O(1+\epsilon)$-approximate Minimum Spanning Tree (MST) on a weighted graph.
For both 
problems we present fully-dynamic deterministic algorithms that run in $O(1)$ rounds per update in the worst case, with $O(\sqrt{N})$ active machines and $O(\sqrt{N})$ total communication per round.
Notice that, in order to maintain the connected components of a graph, it suffices to maintain a spanning forest of the graph.
As it is the case also for centralized algorithms, the hard case is to handle the deletion of edges from the maintained spanning forest.
The main ingredient in our approach is the use of Euler tour of a spanning tree in each connected component. 
This enables us to distinguish 
between different trees of the spanning forest, based on the tour numbers assigned to each of vertices of the trees, which we use to determine  
whether a vertex has an edge to particular part of a tree. 
Notice that to achieve such a bound, each vertex needs to known the appearance numbers of its neighbors in the Euler tour, which one cannot afford to request at each round as this would lead to $O(N)$ communication.
We show how to leverage the properties of the Euler tour in order to avoid this expensive step.
In the static case, the best known algorithm to compute the connected components and the MST of a graph requires $O(c/\delta)$ rounds when $N\in \Omega(n^{1+c})$ and $S\in \Omega(n^{1+\delta})$ \cite{lattanzi2011filtering}.
In the case where $S\in o(n)$, \cite{Chitnis:2013:FCC:2510649.2511220} presented an algorithm to compute the connected components of a graph in $O(\log n)$ rounds, with all the machines and $\Omega(N)$ communication per round.

\paragraph{Bounds from the dynamic algorithms literature.}
We present a reduction to dynamic algorithms in the centralized computational model. More specifically, we show that if there exists a centralized algorithm with update time $u(m,n)$ and preprocessing time $p(m,n)$ on a graph with $m$ edges and $n$ vertices, then there exists a dynamic MPC algorithm which updates the solution in $O(u(m,n))$ rounds with $O(1)$ active machines per round and $O(1)$ total communication, after $p(m,n)$ rounds of preprocessing.
The characteristics of the centralized algorithm (e.g., amortized or worst-case update time, randomized or deterministic) carry over to the MPC model.
This reduction, for instance, implies an amortized $\widetilde{O}(1)$ round fully-dynamic DMPC algorithm for maintaining the connected components or the maximum spanning tree (MST) of a graph \cite{holm2001poly}, and an amortized $O(1)$ round fully-dynamic DMPC algorithm for the maximal matching problem \cite{somon2016fully}. These algorithms however do not guarantee worst-case update time, which is important in applications. Moreover, the connected components and MST algorithms have super-constant round complexity.

\smallskip\noindent\textbf{Road map.} In Section~\ref{sec:model} we introduce the DMPC model. Then, in Sections~\ref{sec:maximal-matching} and \ref{sec:mat1} we present our maximal matching and $\nicefrac{3}{2}$-approximate matching, respectively. 
We present our connected components and $(1+\epsilon)$-approximate MST algorithms in Section~\ref{sec:cc}. 
In Section~\ref{app:mat2}, we present our $(2+\epsilon)$-approximate matching algorithm, and finally the reduction is presented in Section~\ref{app:red}.
\section{The model}\label{sec:model}

In this work we build on the model that was introduced by Karloff, Suri, and Vassilvitski \cite{karloff2010model}, and further refined in \cite{andoni2014parallel,beame2013communication,goodrich2011sorting}. 
This model is commonly referred to as the \emph{Massive Parallel Computing (MPC)} model.
In its abstraction, the MPC model is the following.
The parallel system is composed by a set of $\mu$ machines $M_1, \dots, M_\mu$, each equipped with a memory that fits up to $S$ bits.
The machines exchange messages in synchronous rounds, and each machine can send and receive messages of total size up to $S$ at each round.
The input, of size $N$, is stored across the different machines in an arbitrary way.
We assume that $S,\mu \in O(N^{1-\epsilon})$, for a sufficiently small $\epsilon$.
The computation proceeds in rounds. 
In each round, each machine receives messages from the previous round.
Next, the machine processes the data stored  in its memory without communicating with other machines.
Finally, each machines sends messages to other machines.
At the end of the computation, the output is stored across the different machines and it is outputted collectively. The data output by each
machine has to fit in its local memory and, hence, each machine can output at most $S$ bits.

Since at each round all machines can send and receive messages of total size $S$, the total communication per round is bounded by $S \cdot \mu \in O(N^{2-2\epsilon})$.
See \cite{karloff2010model} for a discussion and justification. 
When designing MPC algorithms, there are three parameters that need to be bounded:

-- Machine Memory: In each round the total memory used by each machine is $O(N^{(1-\epsilon)})$ bits.

-- Total Memory: The total amount of data communicated at any round is $O(N^{(2-2\epsilon)})$ bits.

-- Rounds: The number of rounds is $O(\log^i n)$, for a small $i\geq 0$.

Several problems are known to admit a constant-round algorithm, such as sorting and searching \cite{goodrich2011sorting}.

\smallskip\noindent\textbf{Dynamic algorithms.}
In the centralized model of computation, dynamic algorithms have been extensively studied in the past few decades. The goal of a dynamic algorithm is to maintain the solution to a problem while the input undergoes updates. 
The objective is to update the solution to the latest version of the input, while minimizing the time spent for each update on the input. 
A secondary optimization quantity is the total space  required throughout the whole sequence of updates.

A dynamic graph algorithm is called \emph{incremental} if it allows edge insertions only, \emph{decremental} if it allows edge deletions only, and \emph{fully-dynamic} if it allows an intermixed sequence of both edge insertions and edge deletions. Most basic problems have been studied in the dynamic centralized model, and they admit efficient update times. Some of these problems include, connectivity and minimum spanning tree \cite{holm2001poly, nanongkai2017dynamic}, approximate matching \cite{arar2018dynamic, baswana2011fully, bernstein2019adeamortization, charikar2018fully, neiman2016simple, somon2016fully}, shortest paths \cite{Abraham2017fully, demetrescu2001fully}.

\smallskip\noindent\textbf{Dynamic algorithms in the DMPC model.}
Let $G=(V,E)$ be a graph with $n=|V|$ vertices and $m=|E|$ edges.
In the general setting of the MPC model, where the memory of each machine is strictly sublinear in $n$, 
no algorithms with constant number of rounds are known even for very basic graph problems, 
such as maximal matching, approximate
weighted matching, connected components. 
Recomputing the solution for each of those problems requires  $O(\log n)$ rounds, the amount of data that is shuffled between any two rounds can be as large as $O(N)$, all the machines are active in each round, and all machines need to communicate with each other.
Therefore, it is natural to ask whether we can update the solution to these problems after a small change to the input graph, using a smaller number of rounds, less active machines per round, and less total communication per round. 

Notice that bounding the number of machines that communicate 
immediately implies the same bound on the active machines per round.
For convenience, we call active the machines that are involved in communication. 
The number of active machines also implies a bound on the amount of data that are sent in one round, as each machine has information at most equal to its memory (i.e., $O(\sqrt{N})$ bits). 
The complexity of a dynamic algorithm in the DMPC model can be characterized by the following three factors: 

-- The \emph{number of rounds} required to update the solution. 

-- The \emph{number of machines} that are active per round. 

-- The \emph{total amount of data involved in the communication} per round. 

An ideal algorithm in the DMPC model processes each update using a constant number of rounds, using constant number of machines and constant amount of total communication. While such an algorithm might not always be feasible, a dynamic algorithm should use polynomially (or even exponentially) less resources than it's static counterpart in the MPC model.

\paragraph{Use of a coordinator.}
Distributed systems often host multiple jobs simultaneously, which causes 
different jobs to compete for 
resources.
Additionally, systems relying on many machines to work simultaneously are prone to failures of either machines or channels of communication between the machines.
Our model, allows solutions where all updates are sent to a single (arbitrary, but fixed) machine that keeps additional information on the status of the maintained solution, and then coordinates the rest of the machines to perform the update, by sending them large messages containing the additional information that it stores. 
Examples of such an algorithm is our algorithm for the maximal matching, and the $3/2$ approximate matching. 
In practice, the use of a coordinator might create bottlenecks in the total running time, since it involves transmission of large messages, and also makes the system vulnerable to failures (i.e., if the coordinator fails, one might not be able to recover the solution).

We note that the role of the coordinator in our matching algorithms is not to simulate centralized algorithms (as we do in our reduction from DMPC algorithms to dynamic centralized algorithms), i.e., to perform all computation at the coordinator machine while treating the rest of the machines as memory. 
In particular, we treat the coordinator as a buffer of updates and changes of the solution, and we communicate this buffer to the rest of the machines on a need-to-know basis. 

\paragraph{Algorithmic challenges.} The main algorithmic challenges imposed by our model are the sublinear memory (most of the algorithm known in the MPC model use memory in $\Omega(n)$) and the restriction on the number of machines used in every round. This second point is the main difference between the MPC and DMPC model and poses a set of new interesting challenges.

\section{Fully-dynamic DMPC algorithm for maximal matching}
\label{sec:maximal-matching}
In this section we present a deterministic fully-dynamic algorithm for maintaining a maximal matching with a constant number of rounds per update and a constant worst-case number of active machines per update, when the memory of each machine is $\Omega(\sqrt{N})$ bits, where $N=(m+n)$ 
and $m$ is the maximum number of edges throughout the update sequence. The communication per round is $O(\sqrt{N})$.
Recall that our model introduces additional restrictions in the design of efficient algorithms. 
Specifically, the memory of each machine might not  even be sufficient to store 
the neighborhood of a single vertex, which implies that 
the edges incident to a single vertex may be stored in polynomially many machines.
In this framework, a scan of the neighbors of a single vertex requires a polynomially number of active machines in each round.

Our algorithm borrows an observation from the fully-dynamic algorithm for maximal matching of Neiman and Solomon \cite{neiman2016simple}, which has $O(\sqrt{m})$ worst-case update time and $O(n^2)$ space, 
or the same amortized update bound with $O(m)$ space.
Specifically, Neiman and Solomon \cite{neiman2016simple} observe that a vertex either has a low degree, or has only few neighbors with high degree.
This allows us to treat vertices with large degree separately from those with relatively small degree.
We call a vertex \emph{heavy} if it has a large degree and \emph{light} if it has a small degree.
The threshold in the number of vertices that distinguishes light from heavy vertices is set to be $2\sqrt{m}$.
As the memory of each machine is $\Omega(\sqrt{m})$, we can fit the light vertices together with their edges on a single machine, but for heavy vertices we can keep only up to $O(\sqrt{m})$ of their edges in a single machine.
Given that each vertex knows whether it is an endpoint of a matched edge, the only non-trivial update to be handled is when an edge $e=(x,y)$ of the matching is deleted and we have to check whether there exists an edge adjacent to $x$ or $y$ that can be added to the matching.
Notice that if the neighborhood of each vertex fits in a single machine, then it is trivial to bound the number of rounds to update the solution, as it is sufficient to search for free neighbors of $x$ and $y$ that can be matched to those vertices. 
Such a search can be done in a couple of rounds by sending a message from $x$ and $y$ to their neighbors to ask whether they are free to join or not.
However, this does not immediately bound the number of active machines per round.

\smallskip\noindent\textbf{Overview.}
Our algorithm keeps for each light vertex all the edges of its adjacency list in a single machine. 
For every heavy node we keep only $\sqrt{2m}$ edges that we call \emph{alive}. We call \emph{suspended} the rest of the edges of $v$.
We initially invoke an existing algorithm to compute a maximal matching in $O(\log n)$ rounds. 
Our algorithm always maintains a matching with 
the following invariant: 

\begin{invariant}
	\label{inv:heavy-matched}
No heavy vertex gets unmatched throughout the execution of the algorithm\footnote{After computing the initial maximal matching some heavy vertices might be unmatched. During the update sequence, once a heavy vertex gets matched, it is not being removed from the matching, unless it becomes light again}. 
\end{invariant}

If a new edge gets inserted to the graph, we simply check if we can add it to the matching (i.e., if both its endpoints are free). 
Now assume that an edge $(x,y)$ from the matching gets deleted. 
If both the endpoints are light, then we just scan their adjacency lists (that lie in a single machine) to find a replacement edge for each endpoint of $(x,y)$.
If $x$ is heavy, then we search the $\neighborsbounds$ alive edges of $x$ and if we find a neighbor that is free we match it.
If we cannot find a free neighbor of $x$, then among the (matched) $\neighborsbounds$ alive neighbors of $x$ there should exist a neighbor $w$ with a light mate $z$ (as otherwise the sum of degrees of the mates of neighbors of $x$ would exceed $m$), in which case we remove $(w,z)$ from the matching, we add $(x,w)$ to the matching, and we search the neighbors of the (light) vertex $z$ for a free neighbor to match $z$. 
If $y$ is heavy, we proceed analogously.

We build the necessary machinery in order to keep updated the aforementioned allocation of the adjacency lists to the available machines. This involves moving edges between machines whenever this is necessary, which introduces several challenges, since we cannot maintain updated the information in all machines with only $O(1)$ message exchange. On the other hand, we cannot allocate edges to an arbitrary number of machines. 
We deal with these issues by periodically updating the machines by taking advantage of the fact that we can send large messages from the coordinator machine.

\smallskip\noindent\textbf{Initialization and bookkeeping.}
Our algorithm makes use of $O(\sqrt{N})$ machines. 
We assume that all vertices of the graph contain IDs from $1$ to $n$.
Our algorithm executes the following preprocessing.
First, we compute a maximal matching (this can be done in $O(\log n)$ rounds with the randomized CONGEST algorithm from~\cite{ISRAELI198677}).
Together with each edge in the graph we store whether an endpoint of the edge is matched:  if it is, we also store its mates in the matching.
In a second phase, we compute the degree of each vertex (this can be done in $O(1)$ rounds for all vertices). 
We place the vertices into the machines in such a way that the whole adjacency list of light vertices, and arbitrary $\neighborsbounds$ edges from the adjacency list of heavy vertices, are stored in single machines. 
The remaining adjacency list of a heavy vertex is stored in separate exclusive machines (only store edges of that vertex) so that as few machines as possible are used to store the adjacency list of a heavy vertex.
On the other hand, the light vertices are grouped together into machines. 
The machines that store heavy vertices are characterized as \emph{heavy machines}, and those storing adjacency lists of light vertices as \emph{light machines}.

One of the machines acts as the coordinator, in the sense that all the queries and updates are executed through it. 
The coordinator machine, denoted by $M_C$, stores an update-history $\mathcal{H}$ of the last $O(\sqrt{N})$ updates in both the input and the maintained solution, i.e., which edges have been inserted and deleted from the input  in the last $\sqrt{N}$ updates and which edges have been inserted and deleted in the maintained maximal matching. 
Moreover, for each newly inserted edge that exists in the update-history we store a binary value for each of its endpoints that indicates if their adjacency list has been updated to include the edge.

For convenience, throughout this section we say that the algorithm invokes some function without specifying that all the communication is made through $M_C$.
We dedicate $O(n/\sqrt{N})$ machines to store statistics about the vertices of the graphs, such as their degree, whether they are matched and who is their mate, the machine storing their alive edges, the last in the sequence of machines storing their suspended edges (we treat the machines storing suspended edges as a stack).
To keep track of which machine keeps information about which vertices, we allocate many vertices with consecutive IDs to a single machine so that we can store the range of IDs stored in each machine. 
Hence in $M_C$, except of the update-history $\mathcal{H}$, we also store for each range of vertex IDs the machine that contains their statistics.
This information fits in the memory of $M_C$ as the number of machines is $O(\sqrt{N})$.
Finally, $M_C$ also stores the memory available  in each machine.

\smallskip\noindent\textbf{Maintaining the bookkeeping.}
In what follows, for the sake of simplicity, we assume that the update-history $\mathcal{H}$ is updated automatically. 
Further, we skip the description of the trivial update or queries on the statistics of a vertex, such as its degree, whether it is an endpoint of a matched edge, the machine storing its alive edges, etc.
All of these can be done in $O(1)$ rounds via a message through the coordinator machine $M_C$.
After each update to the graph, we update the information that is stored in a machine by executing those updates in a round-robin fashion, that is, each machine is updated after at most $O(\sqrt{N})$ updates. Recall that we use $O(\sqrt{N})$ machines.

Throughout the sequence of updates we use the following set of supporting procedures to maintain a valid allocation of the vertices into machines: 

-- {\bf $\getalive(x):$} Returns the ID of the machine storing the alive neighbors of $x$.

-- {\bf $\getdegreeinmachine(M,x):$} Returns $x$'s degree in machine $M$. 

-- {\bf $\getsuspended(x):$} Returns the ID of the last in the sequence of heavy machines storing the edges of $x$. 

-- {\bf $\fits(M, s):$} Return $true$ if $s$ edges fit into a light machine $M$, and $false$ otherwise.

-- {\bf $\tofit(s):$} Returns the ID of a light machine that has enough memory to store $s$ edges, and the available space in that machine.	

-- {\bf $\addedge((x,y))$:}
	We only describe the procedure for $x$, as the case for $y$ is completely analogous.
	If $x$ is heavy, add $(x,y)$ in the machine $\getsuspended(x)$ if it fits, or otherwise to a new machine, and set the new machine to be $\getsuspended(x)$.
	If, on the other hand, $x$ is light and $(x,y)$ fits into $\getalive(x)$, 
	we simply add $(x,y)$ in $\getalive(x)$.
	If, $(x,y)$ does not fit in $\getalive(x)$ then call $\moveedges(x,s,M_x,\tofit(s),\mathcal{H})$, where $s$ is the number of alive edges of $x$ (if $x$ becomes heavy, we mark that).
	If all of the remaining edges in the machine $M_x$ (of light vertices other than $x$) fit into another machine, then move them there (this is to bound the number of used machines).

-- {\bf $\moveedges(x, s, M_1, M_2, \mathcal{H})$}, where $x$ is light: First, remove from machine $M_1$ deleted edges of $x$ based on $\mathcal{H}$.
	Second, send from $M_1$ up to $s$ edges of $x$ to $M_2$. 
	If the $s$ edges do not fit into $M_2$, move the neighbors of $x$ from $M_2$ to a machine that fits them, i.e., execute $M_{x'}=\tofit(s+\getdegreeinmachine(M,x))$, move the $s$ edges of $x$ in $M_1$ to $M_{x'}$ and call $\moveedges(x, \getdegreeinmachine(M,x), M_2, M_{x'}, \mathcal{H})$. 

-- {\bf $\fetchmore(x,s)$}, where $x$ is heavy: Moves $s$ suspended edges to the machine $M_x=\getalive(x)$. 
	To achieve this we call \\
	$\moveedges(x, s, \getsuspended(x),  M_x)$.
	While the number of edges moved to $M_x$ is $s'<s$, call  $\moveedges(x, s-s',\getsuspended(x), M_x)$. 

-- {\bf $\movesuspended(x,s,L)$}, where $x$ is heavy: Moves the set $L$ of $s$ edges of $x$ from machine $\getalive(x)$ to the machines storing the suspended edges of $x$. We first fit as many edges as possible in the machine $\getsuspended(x)$, and the rest (if any) to a new machine.

-- {\bf $\updatevertex(x, \mathcal{H}):$} Update the neighbors of $x$ that are stored in $M_x = \getalive(x)$ based on $\mathcal{H}$.
	If $x$ is heavy and the number of edges from the adjacency list of $x$ in $M$ is $s < \neighborsbounds$, then call $\fetchmore(x,\neighborsbounds-s)$.
	If $x$ is heavy and the set of alive edges has size $s>\neighborsbounds$, then call $\movesuspended(x,s-\neighborsbounds,L)$, where $L$ are $s-\neighborsbounds$ edges of $x$  that do not contain the edge $(x,mate(x))$.
	If, on the other hand, $x$ is light and the set of alive edges of $x$ does not fit in $M_x$ after the update, call $\moveedges(x,s,M_x,\tofit(s),\mathcal{H})$, where $s$ is the number of alive edges of $x$.
	If all of the remaining edges in the machine $M_x$ (of light vertices other than $x$) fit into another machine, then move them there (this is to bound the number of used machines). 

-- {\bf $\updateMachine(M, \mathcal{H}):$} Update all adjacency lists stored in machine $M$ to reflect the changes in the update-history $\mathcal{H}$. 
	If $M$ is a heavy machine of a vertex $x$, we proceed as in the case of $\updatevertex(x, \mathcal{H})$, but now on machine $M$ rather than $\getalive(x)$.
	Now we assume $M$ is light.
	First, delete the necessary edges of the light vertices stored at $M$ based on $\mathcal{H}$.  
	If all of the remaining edges of the machine fit into another half-full machine, then move them there (this is to bound the number of used machines).

\smallskip\noindent\textbf{Handling updates.}
Now we describe how our algorithm updates the maintained maximal matching after an edge insertion or an edge deletion.
\paragraph{\bf$insert(x,y)$}
First, execute $\updatevertex(x)$, $\updatevertex(y)$, and $\addedge((x,y))$. 
If both $x$ and $y$ are matched then do nothing and return. 
If neither $x$ nor $y$ are matched, add $(x,y)$ to the matching and return.
In the case where $x$ is matched and heavy and $y$ is unmatched and light then do nothing and return.
The same happens if $y$ is matched and heavy and $x$ is unmatched.
If $x$ is unmatched and heavy, search for a (matched, as this is a maximal matching) neighbor $w$ of $x$ whose mate $z$ is light, remove $(w,z)$ from the matching, add $(x,w)$ to the matching, and if $z$ (who is a light vertex) has an unmatched neighbor $q$ add $(z,q)$ to the matching. If $y$ is unmatched and heavy proceed analogously. Note that this restores Invariant \ref{inv:heavy-matched}. 
In any case, the update-history is updated to reflect all the changes caused by the insertion of $(x,y)$.

\paragraph{\bf$delete(x,y)$}
First, update $\mathcal{H}$ to reflect the deletion of $(x,y)$ and call $\updatevertex(x)$ and $\updatevertex(y)$.
If $(x,y)$ is not in the matching do nothing and return. (The edge has already been deleted from the adjacency lists via the calls to $\updatevertex$.) 
If $(x,y)$ is in the matching proceed as follows.
First, remove $(x,y)$ from the matching.
If $z\in\{x,y\}$ is heavy, search for a neighbor $w$ of $z$ whose mate $w'$ is light, remove $(w,w')$ from the matching, add $(z,w)$ to the matching, and if $w'$ (who is a light vertex) has an unmatched neighbor $q$ add $(w',q)$ to the matching. 
If $z\in\{x,y\}$ is light, scan the neighborhoods of $z$ for a unmatched vertex $w$, and add $(z,w)$ to the matching. 
In any case, the update-history is updated to reflect all the changes caused by the deletion of $(x,y)$.

\begin{lemma}
The algorithm 
uses $O(\sqrt{N})$ machines.
\end{lemma}
\begin{proof}
We show that we maintain at most twice the number of machines than the optimum placement.
Let $M_1, \dots, M_l$ be the machines that store the adjacency list of a heavy vertex $x$, where $M_1=\getalive(x)$. 
Since only $M_l$ is not full, we use at most twice as many machines as the optimum placement for each heavy vertex.
Let now $M_1, \dots, M_l$ be all the machines storing light vertices. 
Since with each update of a light adjacency list we check if we can merge two light machines, it follows that there are no two machines whose edges can be stored in one. 
Hence, our claim holds also in this case.
The lemma follows from the observation that the optimum placement of the edges requires $O(\sqrt{N})$ machines.
\end{proof}

\begin{lemma}
Both $insert(x,y)$ and $delete(x,y)$ run in $O(1)$ rounds, activate $O(1)$ machines per round, and generate $O(\sqrt{N})$ communication per round.
\end{lemma}
\begin{proof}
Recall that we manage the machines that are used to store the sequence of machines storing the suspended edges of heavy vertices as stacks, that is, we store the last machine storing the suspended edges of a vertex $x$ together with the rest of the statistics for $x$, and each machine maintains a pointer to the next machine in the sequence.
Hence, we can access in $O(1)$ rounds the machine that is last in the sequence of machines maintaining the suspended edges of a vertex.
The only supporting function that is not trivially executable in $O(1)$ rounds is  $\fetchmore$.
Note that a call to $\fetchmore$ makes multiple calls to $\moveedges$ to transfer edges suspended edges of a heavy vertex $x$.
As each machine is updated every $O(\sqrt{N})$ rounds, it follows that the number of edges that have been removed from the graph and the machines storing those edges are not yet updated, is $O(\sqrt{N})$.
As all the calls to $\moveedges$ transfer at most $O(\sqrt{N})$ edges of $x$, and all but one machines storing suspended edges of $x$ are full, it follows that there is at most a constant number of calls to $\moveedges$.
\end{proof}

\section{Fully-dynamic 3/2-approximate maximum matching}\label{sec:mat1}

The algorithm for the $3/2$ approximate matching builds on top of the algorithm for maintaining a maximal matching from Section \ref{sec:maximal-matching}. Our algorithm is an adaptation of the algorithm from \cite{neiman2016simple} to our DMPC model.
Our algorithm's approximation is based on a well-known graph-theoretic connection between augmenting paths in an unweighted graph, with respect to a matching, and the approximation factor of the matching relatively to the maximum cardinality matching. 
An augmenting path is a simple path starting and ending at a free vertex, following alternating unmatched and matched edges.
Specifically, a matching that does not contain augmenting paths of length $(2k-1)$ in a graph, is a $(1+\frac{1}{k})$-approximate matching~\cite{hopcroft1973n}.  
In this section we show that it is possible to use the technique in~\cite{neiman2016simple} to design a simple DMPC algorithm for $k=2$.
The additional information that the algorithm needs to maintain, compared to the algorithm from Section \ref{sec:maximal-matching}, is the number of unmatched neighbors of each vertex. 
We call these counters \emph{free-neighbor} counters of the light vertices.
We keep this information in the $O(n/\sqrt{N})$ machines storing the statistics about the vertices of the graph.
In this algorithm, we assume that the computation starts from the empty graphs (An initialization algorithm for this problem would require eliminating all augmenting paths of length 3, but we are not aware of such an algorithm that does not require considerably more than $O(N)$ total memory).

Since the algorithm from Section \ref{sec:maximal-matching} maintains a matching where all heavy vertices are always matched, we only need to update the free-neighbor counters whenever a light vertex changes its matching status. 
Recall that a light vertex keeps all of its neighbors in the same machine. 
Therefore, we simply need to update the counters of the neighbors of the light vertex.
This requires a message of size $O(\sqrt{N})$ from the light vertex $v$ that changed its status to the coordinator and from there appropriate messages of total size $O(\sqrt{N})$ to the $O(n/\sqrt{N})$ machines storing the free-neighbor counters of the neighbors of $v$.

Given that we maintain for each vertex its free-neighbor counter, we can quickly identify whether an edge update introduces augmenting paths of length $3$. 
The modifications to the algorithm from Section \ref{sec:maximal-matching} are as follows.

-- In the case of the insertion of edge $(u,v)$, if $u$ is matched but $v$ unmatched, we check whether the mate $u'$ of $u$ has a free neighbor $w$; if this is the case, we remove $(u,u')$ from the matching and we add $(w,u')$ and $(u,v)$ (this is an augmenting path of length 3). The only free-neighbor counters that we have to update are those of the neighbors of $w$ and $v$, as no other vertices change their status, and no new augmenting paths are introduces as no matched vertex gets unmatched.

-- If both $u$ and $v$ are free after the insertion of $(u,v)$, we add $(u,v)$ into the matching and  update the free-neighbor counters of all neighbors of $u$ and $v$ (who are light vertices, as all heavy vertices are matched).

-- If we delete an edge which is not in the matching, 
	then we simply update the free-neighbor counters of its two endpoints, if necessary.
	
-- Whenever an edge $(u,v)$ of the matching is deleted, we treat $u$ as follows. If $u$ has a free neighbor $w$, then we add $(u,w)$ to the matching and update the free-neighbor counters of the neighbors of $w$ (who is light).
	If $u$ is light but has no free neighbors, then we search for an augmenting path of length 3 starting from $u$. 
	To do so, it is sufficient to identify a neighbor $w$ of $u$ whose mate $w'$ has a free neighbor $z\not=u$.
	If there exists such $w'$ then we remove $(w,w')$ from the matching and add $(u,w)$ and $(w',z)$ to the matching, and finally update the free-neighbor counters of the neighbors of $z$ (who is light). No other vertex changes its status.
	If on the other hand, $u$ is heavy, then we find an alive neighbor $w$ of $u$ with a light mate $w'$, remove $(w,w')$ from the matching and add $(u,w)$ to it.
	(This can be done in $O(1)$ rounds communication through the coordinator with the, up to $n/\sqrt{N}$, machines storing the mates of the statistics of the $O(\sqrt{N})$ alive neighbors of $w'$.)
	Finally, given that $w'$ is light, we proceed as before trying to either match $w'$ or find an augmenting path of length 3 starting from $w'$.
	Then, we proceed similarly to the case where $u$ was light.

Notice that in all cases where we have to update the free-neighbor counters of all neighbors of a vertex $v$, $v$ is a light vertex, so there are at most $O(\sqrt{N})$ counters to be updated and thus they can be accessed in $O(1)$ rounds, using $O(n/\sqrt{N})$ active machines, and $O(\sqrt{N})$ communication complexity. Hence, given the guarantees from Section \ref{sec:maximal-matching} and the fact that we only take a constant number of actions per edge insertion or deletion, we conclude that our algorithm updates the maintained matching  
in $O(1)$ rounds, using $O(n/\sqrt{N})$ machines and $O(\sqrt{N})$ communication per round in the worst case.
We conclude this section by proving the approximation factor of our algorithm.

\begin{lemma}
The algorithm described in this section correctly maintains a $3/2$-approximate matching.
\end{lemma}
\begin{proof}
In order to guarantee the $3/2$ approximation we need to argue that there are no augmenting paths of length $3$ or more.
Such a path exists if and only if there is an edge of the matching whose both endpoints have a free neighbor.
We show that after every update made by our algorithm, we eliminate all such  matched edges.
That is, for each edge of the matching we ensure that at most one endpoint has a free neighbor.
We proceed with a case study, assuming that our claim holds just before the update we consider.
Recall that the maintained matching is always maximal as we build on the algorithm from Section \ref{sec:maximal-matching}.
The only two cases where we need to search for an augmenting path of length $3$ is when either a new vertex becomes free, or when we connected a matched vertex with a free vertex.
In the case where a vertex $u$ becomes free due to an edge deletion, our algorithm tests whether $u$ is an endpoint of a length-3 augmenting path $\langle u,w,w',z \rangle$, where $w$ is a matched neighbor of $u$ that the mate of $u$ and $w$ a free neighbor of $u'$, if $u'$ has free neighbor, and by removing $(u,u')$ and adding $(u,v)$ and $(u',w)$ to the matching to augment the length $3$ augmenting path. This does not create new augmenting paths as $u$ and $z$ have no free neighbors 
and no new vertex becomes free. 
For the second case where we connect a matched and a free edge, we again search and augment possible augmenting paths of length 3.
Given that all free-neighbors counters are updated every time a vertex enters/leaves the matching, our algorithm maintains a $3/2$-approximate matching.
\end{proof}
\section{Fully-dynamic connected components and approximate MST}\label{sec:cc}

In this section we present a fully-dynamic deterministic distributed algorithm for maintaining the connected components of a graph with constant number of rounds per edge insertion or deletion, in the worst case\footnote{Note that no constant round algorithm for connected component is known for the static case. On the downside, the number of active machines per round is not bounded. We leave as an interesting area of future work to design an algorithm that uses a smaller number of machines per update}.
At the heart of our approach we use Euler tours, which have been successfully used in the design of dynamic connectivity algorithms in the centralized model of computation, e.g., in \cite{henzinger1999randomized,holm2001poly}.
Given a rooted tree $T$ of an undirected graph $G$, an \emph{Euler tour} (in short, E-tour) of $T$ is a path along $T$ that begins at the root and ends at the root, traversing each edge exactly twice. The E-tour is represented by the sequence of the endpoints of the traversed edges, that is, if the path uses the edges $(u,v),(v,w)$, then $v$ appears twice. 
As an E-tour is defined on a tree $T$, we refer to the tree $T$ of an E-tour as the Euler tree (E-tree, in short) of the E-tour.
The root of the E-tree appears as the first and as the last vertex of its E-Tour.
The length of a tour of an E-tree $T$ is $ELength_{T} = 4(|T|-1)$, the endpoints of each edge appear twice in the E-tour.
See Figures \ref{fig:EulerTour-insertion} and \ref{fig:EulerTour-deletion}   for examples.
As the preprocessing shares similarities with the edge insertion, we postpone the description of the preprocessing after describing the update procedure to restore the E-tour after an edge insertion or deletion.

\begin{figure*}
	\vspace{-0.3cm}
	\begin{center}
		\includegraphics[trim={1cm 9cm 2cm 1cm}, clip=true, width=0.88\textwidth]{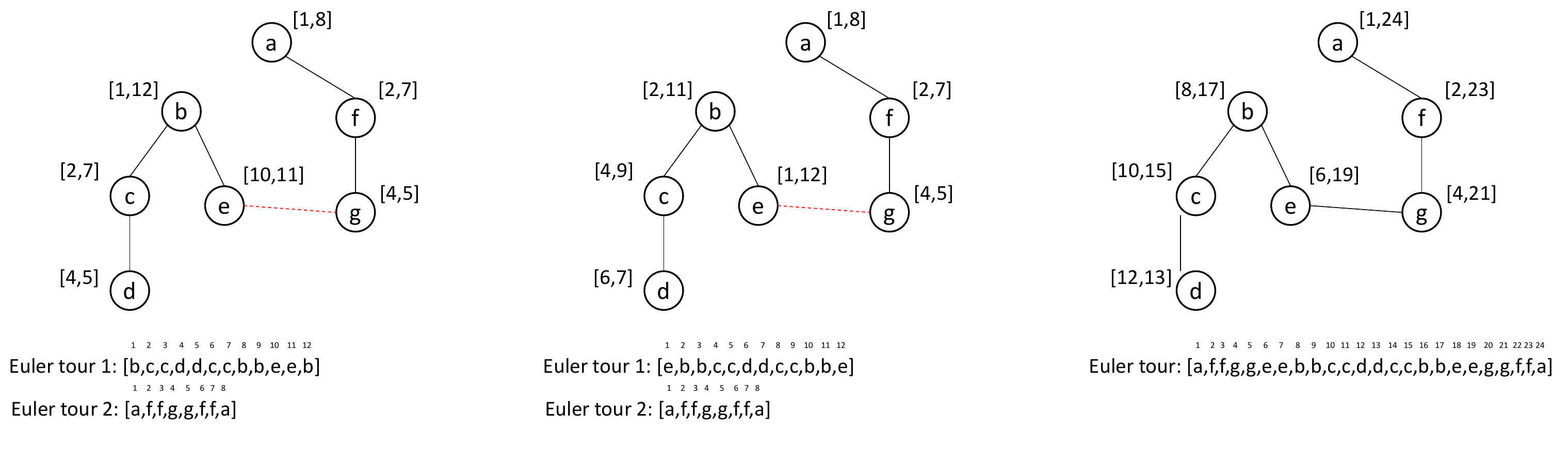}
	\end{center}
	\vspace{-0.6cm}
	\caption{(i) A forest and an E-tour of each of its tree below. The brackets represent the first and the last appearance of a vertex in the E-tour. (ii) The E-tour after setting $e$ to be the root of its tree. (iii) The E-tour after the insertion of the edge $(e,g)$.}
	\label{fig:EulerTour-insertion}
\end{figure*}

\begin{figure*}[t!]
	\vspace{-0.3cm}
	\begin{center}
		\includegraphics[trim={0.8cm 9cm 5cm 1cm}, clip=true, width=0.88\textwidth]{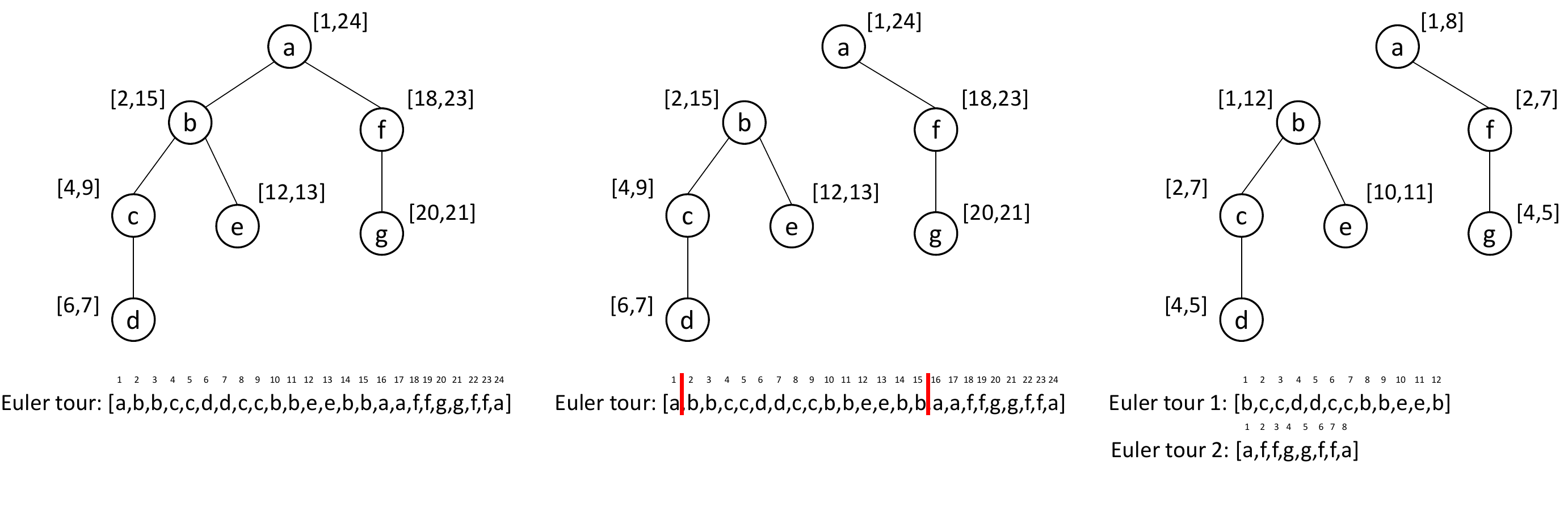}
	\end{center}
	\vspace{-0.6cm}
	\caption{(i) A tree and an E-tour of the tree below it. The brackets represent the first and the last appearance of a vertex in the E-tour. (ii) An intermediate step of the update of the E-tour after the deletion of the edge $(a,b)$. The red lines in the E-tour indicate the split points of outdated E-tour. (iii) The E-tour after the deletion of the edge $(a,b)$.}
	\label{fig:EulerTour-deletion}
	
	\vspace{-0.2cm}
\end{figure*}

We assume that just before an edge update, we maintain for each connected component of the graph a spanning tree, and an E-tour of the spanning tree.
Using vertex-based partitioning we distribute the edges across machines, and each vertex is aware of the ID of its component, and together with each of its edges we maintain the ID of the component that it belongs to and the two indexes in the E-tour (of the tree of the component) that are associated with the edge. 
Moreover, we maintain with each vertex $v$ the index of its first and last appearance in the E-tour of its E-tree, which we denote by $f(v)$ and $l(v)$.
We denote by $index_v$ the set of all indexes that $v$ appears in the E-tour of $T$. 
Note that $|index_v| = 2\cdot d_T(v)$ 
in the E-tour, where $d_T(v)$ is the degree of $v$ in the corresponding E-tree $T$.
We do not explicitly store $index_v$, this is implicitly stored with each vertex as information on $v$'s edges.
Therefore, we perform updates on the indexes in $index_v$ but it is actually the indexes that are stored at the edges that are updated.
To update this information in a distributed fashion, we leverage the properties of an E-tour which allows us to change the root of an E-tree, merge two E-trees, and split an E-tree, by simply communicating the first and last indexes of the new root, or the endpoints of the inserted/deleted edge. 

\smallskip\noindent\textbf{Handling updates.}
The main idea to handle updates efficiently is that the E-tour of the spanning trees can be updated efficiently without requiring a lot of communication. For instance, one can change the root of an E-tree, and update all the information stored in the vertices of that tree, by sending $O(1)$-size messages to all vertices. Moreover, we can test whether a vertex $u$ is an ancestor of a vertex $v$, in their common E-tree, using only the values $f(u)$ and $l(u)$ and $f(v)$ and $l(v)$.
The insertions and deletions of edges in the graph are handled as follows.

\paragraph{$insert(x,y)$:} 
If $x$ and $y$ are in the same connected component, we simply add the edge to the graph. Otherwise, we proceed as follows. 
We first make $y$ the root of its E-tree $T_y$ (if it is not already). 
Let $ELength_{T_y} = 4(|T_y|-1)$ denote the length of the E-tour of $T_y$. 
For each vertex $w$ in $T_y$ we update each index $i
\in index_w$ to be $i = ((i + ELength_{T_y}-l(y))\mod ELength_{T_y}) + 1$. 
These shifts of the indexes of $w$ correspond to a new E-tour starting with the edge between $y$ and its parent, where the parent is defined with respect to the previous root of $T_y$. 
Second, we update the indexes $i\in index_w$ of the vertices $w \in T_y$ to appear after the first appearance of $x$ in the new E-tour. 
For each vertex $w$ in $T_y$ update each index $i\in index_w$ to be $i = i + f(x) + 2$. 
Third, set $index_y = index_y \cup \{f(x)+2, f(x)+l(y)+3\}$ and 
$index_x = index_x \cup \{f(x)+1, f(x)+l(y)+4\}$, where $l(y)$ is the largest index of $y$ in the E-tour of $T_y$ before the insertion of $(x,y)$. 
Finally, to update the indexes of the remaining vertices in $T_x$, for each $i\in index_w$ where $i>f(x)$ we set $i=i+4\cdot ELength_{T_y}$. 
See Figure \ref{fig:EulerTour-insertion} for an illustration.

Notice that the only information required by each vertex $w$ to perform this update, besides $index_w$ which is implicitly stored on the edges of $w$ and $f(w)$, is $ELength_{T_y}, f(y), l(y), f(x), l(x)$, which can be sent to all machines via a constant size message from $x$ and $y$ to all other machines.
Notice that $x$ and $y$ do not need to store $f(x),l(x)$ and $f(y),l(y), ELength_{T_y}$, respectively, as they can simply learn those by sending and receiving an appropriate message to all machines. Hence each insertion can be executed in $O(1)$ rounds using all machines and $O(\sqrt{N})$ total communication per round (as all communication is between $x$ or $y$ with all other machines, and contains messages of $O(1)$ size).

\paragraph{$delete(x,y)$:} 
If $(x,y)$ is not a tree edge 
in the maintained forest, we simply remove the edge from the graph.
Otherwise, we first split the E-tree containing $x$ and $y$ into two E-trees, and then we reconnect it if we find an alternative edge between the two E-trees.
We do that as follows.
We check whether $x$ is an ancestor of $y$ or vice versa by checking whether $f(x)<f(y)$ and $l(x)>l(y)$.
Assume w.l.o.g. that $x$ is an ancestor of $y$ in $T_x$.
First, we set $index_x = index_x \setminus \{ f(y)-1, l(y)+1\}$ and $index_x = index_y \setminus \{ f(y), l(y)\}$ (that is, we simply drop the edge $(x,y)$).
Then, for all descendants $w$ of $y$ in $T_y$ (including $y$), for each $i\in index_w$ set $i = i - f(y)$, where $f(y)$ is the smallest index of $y$ before the deletion of $(x,y)$.
Update $|T_y|$ and allocate a new ID for the new connected component containing $y$.
Second, for all vertices $w \in T_x \setminus T_y$ and all $i\in index_w$ if $i>l(y)$ set $i=i-((l(y)-f(y)+1)+2)$, where $l(y)$ and $f(y)$ are the largest and smallest, respectively, index of $y$ before the deletion of $(x,y)$.
This is to inform all vertices that appear after $l(y)$ that the subtree rooted at $y$ has been removed, and hence the E-tour just cuts them off (the +2 term accounts for the two appearances of x in the E-tour because of $(x,y)$).
Finally, we find an edge from a vertex $v \in T_y$ to a vertex $w\in T_x$, and execute $insert(x,y)$.

Similarly to the case of an edge insertion, all of the above operations can be executed in a constant number of rounds as the only information that is required by the vertices are the ID of the components of $x$ and $y$, and the values $f(x),l(x),f(y),l(y)$, which are sent to all machines.
Moreover, the search of a replacement edge to reconnect the two trees of $x$ and $y$ can be done in $O(1)$ rounds as we only need to send the IDs of the two components to all machines, and then each machine reports at most one edge between these two components to a specific machine (specified also in the initial message to all machines).

\smallskip\noindent\textbf{Preprocessing.}
During the preprocessing, we compute a spanning forest $\mathcal{T}$ of the input graph and an E-tour on each tree $T \in \mathcal{T}$ with arbitrary roots.
We build on top of the $O(\log n)$ randomized algorithm that computes a spanning forest of a graph by iteratively identifying small neighborhoods to contract into single vertices and at each iteration reduces the number of vertices by a constant fraction~\cite{ahn2012analyzing}. 
It is instructive to view the contraction process as merges of connected component that are build-up throughout the execution, where initially all components are singleton vertices.
We augment this algorithm to maintain a spanning tree in each component, as well as an E-tour of each spanning tree.
We do this as follows.
Using vertex-based partitioning we distribute the edges across machines, and each vertex is aware of the ID of its component, and together with each of its edges we maintain the two indexes in the E-tour (of the tree of the component) that are associated with the edge as well as the ID of the component containing the edge. 
At each iteration, several components might merge into one, but all such merges have a common component to which they are contracted; we call this component the \emph{central component} of the merge. 
Whenever two, or more, components merge into one, they all get the ID of the central component.
Each of the components that merge to the central component uses a single edge to merger their spanning tree as a subtree of the spanning tree of the central component.
Let $C_1, C_2, \dots, C_l$ be the components that merge and w.l.o.g. let $C_1$ be the central component of the merge. 
Moreover, let $e_2, \dots, e_l$ be the edges the non-central components use to connect to the central component $C_1$.
Our plan is to simulate the sequence of edge insertions of $e_2, \dots, e_l$ within a constant number of rounds.

First, in parallel for each component $C_i \in \{C_2, \dots, C_l\}$ with connecting edge $e_i = (x_i,y_i), x_i \in C_1, y_i = C_i$, we set its root to be $y_i$. 
This is, essentially, the first step of the insert $e_i$ operation.
Second, we store the tree edges of all components $C_1,\dots,C_l$ into $O(\sqrt{N})$ machines, and we sort them based on the smallest of the indexes of their endpoints. 
(Sorting can be done in $O(1)$ rounds as shown in \cite{goodrich2011sorting}.)
The sorting implies an order of the machines storing the ordered edges; let $M_1, \dots, M_{q}$ be these machines.
For each component $C_i$ with connection edge $e_i=(x_i,y_i), x_i \in C_1, y_i = C_i$, we send the size of the E-tour of $C_i$ (which is $4(|C_i|-1)$), to the machine (among the machines $M_1, \dots, M_{q}$) storing the index $f(x_i)$ and we associate it with that index (it can be part of the message).
If more than one trees connect to the same vertex, we impose a total order among them defined by the IDs of the other endpoints of the connection edges of the components, and for each component $C_j$ in this order, we compute the sum $\psi(C_j)$ of sizes of the components before $C_j$ in that order. If there is a single component $C_j$ connecting to a vertex, then its $\psi(C_j)=0$. (The $\psi$ values are used in the final step of each iteration.)
Within each machine $M_i, 1 \leq i \leq q$ we sum the sizes that were sent to indexes stored at $M_i$ in the previous step, and we send this information to all machines $M_j, i<j\leq q$.
(Each machine sends a message of constant size to each other machine. Hence, all messages can be sent in one round.)
In tandem, we also sum the values on the messages, of the same type, that are received at machine $M_i$ from machines $M_p, 1\leq p < i$.
Now we can compute for each index $i$ of an edge $e=(w,z)$ in $C_i$ the sum of sizes $\phi(i)$ of components that are attached as subtrees to vertices $w$ with smaller value $f(w)<f(v)$ (here we also consider those components that were attached to indexes stored in $M_i$).
This allows use to execute the final step of the procedure of inserting an edge in parallel for all conflicting component merges.
In particular, for each index $j$, we set $j=j+4 \phi(j)$.
Finally, we execute the second step of the process of inserting an edge. 
That is, for each component $C_i, i>1$ with connection edge $e_i=(x_i,y_i), x_i \in C_1, y_i = C_i$, and each index $j$ of an edge in $C_i$ we set $j=j + f(x_i) + 4 \phi(x_i) + 4 \psi(C_i) + 2$.
All additional steps of the base algorithm can be executed in $O(1)$ rounds, and hence the whole preprocessing can be executed in $O(\log n)$ rounds. 

\subsection{Extending to $(1+\epsilon)$-approximate MST}
To maintain a minimum spanning tree instead of a spanning tree, we use the dynamic spanning tree algorithm with the following two changes.
First, whenever an edge $(x,y)$ is added and the two endpoints are already in the same tree, we compute the edge $(u,v)$ with the maximum weight among all the edges whose both endpoints are ancestors of either $x$ or $y$ (but not both) and we compare it to the weight of $(x,y)$ (these tests can be done efficiently using the E-tree). We only keep the edge with the minimum weight among $(u,v)$ and $(x,y)$.
Second, at Step 3 of the $delete$ operation, instead of adding any edge between the two trees, the algorithm adds the minimum among all such edges.

The preprocessing can be adjusted to compute a $(1+\epsilon)$-approximate MST by doing bucketization, which introduces only a $O(\log n)$ factor in the number of rounds. In fact, it is enough to bucket the edges by weights and compute connected components by considering the edges in bucket of increasing weights iteratively and separately.

\section{Maintaining a (2+$\epsilon$)-approximate matching}\label{app:mat2}

In this section we adapt the algorithm by Charikar and Solomon~\cite{charikar2018fully} to get a $(2+\epsilon)$-approximate matching algorithm with $O(1)$ number of rounds per update, $\widetilde{O}(1)$ communication per round, and $\widetilde{O}(1)$ active machines per round.
Although the algorithm from \cite{charikar2018fully} needs small modifications to achieve the above bounds in our model, these are essential as the original algorithm relies on the fact that it is executed sequentially. 
We first give an overview of the algorithm, and then show how one can resolve the issues that arise from the distributed implementation of the algorithm. 

\subsection{Overview of the sequential algorithm}
The algorithm by Charikar and Solomon \cite{charikar2018fully} builds on the framework established by Baswana, Gupta, and Sen~\cite{baswana2011fully} that was designed for fully-dynamic maximal matching with $O(\log n)$ amortized update time.
For ease of presentation, we first very briefly describe the framework from \cite{baswana2011fully} and then the modified version in \cite{charikar2018fully}.
The set of vertices is decomposed into $\log_{\gamma}n+2$ levels, $\gamma \in O(\log n)$.
The unmatched vertices are assigned level $-1$, while the matched vertices are assigned to levels $[0, \dots, \log_{\gamma}n]$, where $\gamma = \theta(n)$.
Denote the level of a vertex $v$ as \lvl{v}.
Let $v$ be a matched vertex and $e=(u,v)$ be the edge of the matching that is adjacent to $v$.
Roughly speaking, the level of $v$ in the level-decomposition is the logarithm (with base $\gamma$) of the cardinality of the sampling space from which $e$ was selected uniformly at random, that is, the fact that $\lvl{v}=\ell$ implies that $e$ was selected uniformly at random among at least $\gamma^{\ell}$ edges.
We refer to the cardinality of the sampling space from which an edge $e$ was selected as the \emph{support} of $e$.
Notice that while neighbors of a vertex $v$ get deleted the support of the edge $e$ reduces, but insertions of new neighbors of $v$ do not increase the support of $e$ as they were not an option when $e$ was sampled.
The aforementioned grouping of the vertices and their adjacent matched edges serves as an estimation of the number of updates needed to delete an edge of the matching at each level. That is, a matching edge at level $\ell$ is expected to be deleted after, roughly, $\gamma^{\ell}/2$ deletions of edges adjacent to $v$.
Moreover, the algorithm maintains an orientation of the edges in the graph where each edge between two vertices $u$ and $v$ is oriented from the vertex with higher level to the vertex of lower level; ties are broken suitably by the algorithm.
The outgoing edges from a vertex $v$ are stored in a list $\outset{v}$, while for the incoming edges of a vertex the algorithm maintains a partition of the edges into lists based on their level, that is, the incoming edges of $v$ at level $\ell \geq \lvl{v}$ are stored in $\inset{v}{\ell}$.
Notice that the more refined maintenance of the incoming edges of a vertex allows vertex $v$ to traverse only the incoming edges at a specific level, while such a process for the outgoing edges requires the traversal of the whole list $\outset{v}$.
At this point it is useful to define the quantity $\lowneighbors{v}{\ell}$ which represents the number of neighbors of vertex $v$ at levels $1$ through $\ell-1$. This is mainly used in the algorithm in \cite{charikar2018fully}.

The algorithm maintains the following invariants during its execution.
\begin{itemize}
    \item[(i)] Any matched vertex has level at least 0.
    \item[(ii)] The endpoints of any matched edge are at the same level, and this level remains unchanged until the edge is deleted from the matching. 
    \item[(iii)] All free vertices have level -1 and out-degree 0. (This guarantees that the matching is maximal.) 
    \item[(iv)] An edge $(u,v)$ with $\lvl{u} > \lvl{v}$ is oriented by the algorithm from $u$ to $v$. In case where $\lvl{u} = \lvl{v}$, the orientation is determined suitably by the algorithm.
\end{itemize}

Whenever an edge is deleted from the matching, the algorithm places each endpoint of the deleted edge at a level $\ell$ such that it can pick an incident matching edge among a pool of $\gamma^{\ell}$ candidate edges. 
We avoid giving the details on how to update the maintained data structures after an edge insertion or deletion, as these details are out of the scope of this paper. 
Roughly speaking, the main idea of the analysis in \cite{baswana2011fully} builds on the fact that to remove a matching edge $e=(u,v)$ at level $\ell$, the adversary needs to delete $O(\gamma^{\ell})$ many edges, which allows the algorithm to accumulate enough potential to restore the maintained invariants by reassigning levels to $u$ and $v$ and update the data structures $\outset{\cdot}$ and $\inset{\cdot}{\cdot}$ of $u$ and $v$ and their affected neighbors.
The bottleneck of the algorithm is in maintaining the data structures $\outset{\cdot}$ and $\inset{\cdot}{\cdot}$ throughout the updates.
With our model, each machine contains local computational power and can send messages in batches to the neighbors of a vertex stored at the machine.
This allows one to update the affected data structures in batches, by simply sending and receiving the appropriate information from each endpoint of the inserted/deleted edge to their neighbors and each individual vertex updates the data structures concerning themselves.
That is, if a vertex changes level, it can learn it can update all the relevant data structure in $O(1)$ rounds using a number of machines and communication that is analogous to the number of neighbors of the vertex.

The algorithm from \cite{charikar2018fully} maintains a relaxed version of the invariants that are maintained by \cite{baswana2011fully}.
As the authors argue themselves, in order for the algorithm to have a subpolynomial worst-case update time it is necessary to be proactive with respect to deletions of matched edges.
More specifically, the authors present a scenario where the adversary can force the algorithm to reduce drastically the support of many edges of the matching, and then remove many edges of the matching that have reduced support, which forces the algorithm to perform a polynomial time computation within few updates.
Charikar and Solomon \cite{charikar2018fully} deal with such situations by designing an algorithm that ensures that at any point in time every edge of the matching at level $\ell$ is sampled uniformly at random from a relatively large sample space (i.e., $\Omega((1-\epsilon)\cdot \gamma^{\ell})$).
That is, the algorithm maintains a relatively large support for each edge of the matching independently of the adversarial deletions.
This is done to keep low the probability of the adversary ``hitting'' an edge of the matching at level $\ell$, and thus, at any point in time only few edges might be deleted from the matching by the adversary.

As hinted in the discussion of the algorithm from \cite{baswana2011fully}, a deletion of an edge from the matching at level $\ell$ can trigger an update that requires $\Omega(\gamma^{\ell})$ time in order to place the endpoints of the deleted edge in the right level and try to match them with another vertex in their neighborhood.
The algorithm from \cite{charikar2018fully} uses a similar approach, with the difference that each long update is executed in small batches of operations, where each batch is executed after an edge update and performs a polylogarithmic number of operations. 
More specifically, each batch contains either $\Delta = O(\log^5n)$ or $\Delta' = \Delta \cdot \log n$ operations, depending on the type of update that is being performed.
In other words, a long process is simulated over polynomially many adversarial edge insertions/deletions.
The period during which the algorithm remains active after an edge insertion or deletion is called \emph{update cycle}.
At a high level, the algorithm ensures a low-enough probability of deleting an edge of the matching which, in turn, allows it to process such a deletion in many update cycles, without having to deal with many such deleted edges simultaneously, with high probability. 
This is achieved by proactively deleting edges of the matching that have low support and then trying to match again the newly free endpoints of the deleted edges; the endpoints of deleted edges by the algorithm are called \emph{temporarily free} vertices.
In addition, to ensure low-enough probability of an adversarial deletion of a matched edge, the algorithm runs several procedures that remain active throughout the sequence of the edge insertions/deletions (one of which  keeps removing edges with low support).
These procedures are called \emph{schedulers}, and each such scheduler is responsible for ensuring different invariants that are maintained throughout the algorithm.
The algorithm executes for each level $-1,0,\dots,\log_{\gamma}n$ of level-decomposition a copy of a scheduler from each type. Each of those copies is called a \emph{subscheduler}, and all subschedulers of the same type are managed from the same scheduler. 
Hence, there are $O(\log_{\gamma}n)$ subschedulers managed by a constant number of schedulers.

Since long updates are executed in many small batches, it is unavoidable that at each round there exist vertices that are in the process of being updated. 
These vertices are called \emph{active} and they are maintained in a list throughout the execution of the algorithm; we call this list the \emph{active list} and denote is by $\mathcal{A}$. 
It is shown that at any point in time there are at most $O(\log n)$ active vertices, with high probability.
The algorithm also maintains the vertices that become free due to adversarial edge deletions. Such vertices are maintained in lists based on the level of the deleted edges, i.e., the algorithms maintains a list $Q_{i}$ at each level $i$.
Recall that the algorithm breaks down the execution of each process in batches of size $\Delta$ or $\Delta'=\Delta\cdot\log n$.
The size of each batch depends on the procedure that initiated the execution and not on the level of the subscheduler that runs the batch; that is, for the batches handled by the same scheduler is uniform across the different levels.
Hence, the execution of a procedure by a subscheduler at level $\ell$, which requires $T_\ell$ time, is carried out in $T_\ell / \Hat{\Delta}$, where $\Hat{\Delta}\in \{\Delta, \Delta'\}$. 
The level-decomposition ensures that a procedure that is executed by a subscheduler at level $\ell$ requires at least as many rounds as any process at levels $\ell' < \ell$.
As a consequence, during an execution of a process at level $\ell$, possibly many processes at lower levels are executed.

Before describing the schedulers that are used by the algorithm, we first review three supporting procedures.
In what follows, similarly to \cite{charikar2018fully}, we assume that the length of the update sequence is limited to $O(n^2)$, and that the maintained matching has size $\Omega(\log^5 n/ \epsilon^4)$. 
These assumptions can be easily removed. 

\paragraph{The authentication process.}
While updating the data structures of a vertex $v$, some of its neighbors potentially change their level multiple times. 
This happens because a procedure handling a vertex at level $\ell$ compared to a procedure handling a vertex at level $\ell' < \ell$ takes $\gamma^{\ell - \ell'}$ times more time (as the exact difference depends on the type of the processes being carried out).
Hence, at the end of the execution of the process handling vertex $v$, vertex $v$ might not be updated about the level of some of its neighbors.
To deal with this, the algorithm keeps track of the neighbors of $v$ that change their level, and acts upon such changes. 
This is implemented efficiently as follows. 
At the end of the execution of a procedure handling a vertex $v$, the algorithm iterates over the list of active vertices and for each active vertex $z$ the information of $v$ about $z$ is being updated.
Since two procedures might have a very different execution times, we also need to take care of the scenario where a neighbor $w$ of $v$ enters and leaves the active list $\mathcal{A}$ before $v$ finishes its execution. 
However, just before $w$ leaves $\mathcal{A}$, both $v$ and $w$ are active, and hence, it suffices to scan the active list $\mathcal{A}$ and for each neighbor $z$ of $w$ that is active (that includes $v$), update their information about $w$.
Since $|\mathcal{A}| = O(\log n)$, searching for all neighbors of a vertex in the list and updating their mutual information takes $O(\log^2 n)$ time, which means that it can be executed in a single batch (i.e., it should not be simulated in multiple update rounds).
In our model, this can be implemented in $O(1)$ rounds using only $\widetilde{O}(1)$ machines per round.

\paragraph*{Procedure $\setlevel(v,\ell)$.}
This supporting procedure is responsible to set the level of $v$ to be $\ell$ and to update all the necessary data structures of $v$ and its affected neighbors. This procedure is called by the various subschedulers to facilitate the level change that is associated with them. Notice that the level $\ell$ to which $v$ is set is not determined by the procedure itself, but by its caller.
We refer to this process as the rising, or falling, of $v$ depending on whether $\lvl{v}<\ell$ or $\lvl{v}>\ell$, respectively, where $\lvl{v}$ is the level of $v$ before the call of \setlevel{} procedure.
This process is executed by a level $\hat{\ell}=\max\{\ell,\lvl{v}\}$ subscheduler.
The procedure runs in a total of $O(\gamma^{\hat{\ell}})$ time, which is executed in batches of size $\Delta$ or $\Delta'$ (depending on the subscheduler calling it).

The procedure starts by storing the old level of $v$ (that is, $\ell_v^{old} = \lvl{v}$), and setting $\lvl{v} = \ell$.
Then it updates the vertices in $\outset{v}$ about the new level of $v$, that is, for each vertex $w\in \outset{v}$ such that $\lvl{w}<\ell$ it moves $v$ from $\inset{w}{\ell_v^{old}}$ to $\inset{w}{\ell}$.
Next, depending on whether $v$ is rising of falling, we further need to flip the outgoing (resp., incoming) edges of $v$ with its appropriate neighbors to restore the invariants of the level-decomposition. 
In the case where $v$ is falling, that is $\ell < \ell^{old}_v$, for each vertex $w\in \outset{v}$ such that $\ell<\lvl{w}\leq \ell_v^{old}$ we move $w$ from $\outset{v}$ to $\inset{v}{\lvl{w}}$ and $v$ from set $\inset{w}{\ell_{v}^{old}}$ to $\outset{w}$.
We further need to update the value $\lowneighbors{w}{i}$, 
for all $w\in \outset{v}$ and all $\ell+1 \leq i \leq \ell_{v}^{old}$.
We again do that by iterating through the set $\outset{v}$ and for each edge we increment all appropriate counters.
The procedure is analogous for the case where $v$ is rising.

Recall that the $O(\gamma^{\Hat{\ell}})$ time required by  procedure \setlevel, to change the level of vertex $v$ from $\ell_{v}^{old}$ to $\ell$ where $\hat{\ell} = \max \{\ell, \ell_{v}^{old}\}$, is executed in batches of size $\Hat{\Delta}\in \{\Delta, \Delta'\}$.
In our distributed implementation of the algorithm we will execute all $\Delta$ operations of each batch of procedure \setlevel{} in one round.
To do so, we notice that all updates in the data structures of $v$ and its neighbors are independent from each other, that is, the updates made in the data structure of each neighbors $w\in \outset{v}$ do not depend on the preceding or succeeding updates to other neighbors of $v$. Hence, we can execute all of them in parallel. 
We can use the machinery developed in Section \ref{sec:maximal-matching} to identify to which machine each message should be delivered.
In other words, the $\Hat{\Delta}$ operations that are executed by each subscheduler at each update round are performed in $O(1)$ MPC rounds.

\paragraph*{Procedure \handlefree$(v)$.}
This procedure is responsible for handling a temporarily free vertex $v$. The procedure first identifies the highest level $\ell$, such that $\lowneighbors{v}{\ell} \geq \gamma^\ell$ (recall that $\lowneighbors{v}{\ell}$ is the number of neighbors of $v$ in level strictly lower than $\ell$), and then the corresponding set $S(v)$ of non-active neighbors of $v$ at levels lower than $\ell$. 
Then the procedure samples uniformly at random a vertex $w$ from $S(v)\setminus \mathcal{A}$ as the new mate of $v$ in the matching. 
To match $v$ with $w$ the procedure does the following. It first unmatches $w$ from its former mate $w'$, then $v$ and $w$ are set to level $\ell$ using calls to \setlevel$(v,\ell)$ and \setlevel$(w,\ell)$ and adds edge $(w,v)$ into the matching. Finally, if $w$ was previously matched with $w'$ the procedure \handlefree$(w')$ is called recursively.
The running time of \handlefree$(v)$ is bounded by $O(\gamma^{\lvl{v}}\log^4n)$ (see \cite{charikar2018fully} for the analysis), and it is executed in batches of size $\Hat{\Delta}$.
Note that also in this case we can execute all $\Hat{\Delta}$ operations in a constant number of rounds.

\paragraph*{Maintained invariants.}

The algorithm in \cite{charikar2018fully} maintains the following invariants: 
\begin{itemize}
    \item[(a)] Any matched vertex has level at least 0. 
    \item[(b)] The endpoints of any matched edge are of the same level, and this level remains unchanged until the edge is deleted from the matching. (This defines the level of a matched edge, which is at least 0 by item (a), as the level of its endpoints.)
    \item[(c)] Any vertex of level -1 is unmatched and has out-degree 0.
    \item[(d)] An edge $(u, v)$, such that $\lvl{u} > \lvl{v}$ and $u$ and $v$ 
    are not temporarily free, is oriented as from $u$ to $v$.
    \item[(e)] For any level-$\ell$ matched edges $e$ with $T_\ell/\Delta \geq 1$ and any $t$, it holds that $|S_t(e)| > (1-2\epsilon)\cdot \gamma^{\ell}$.
    \item[(f)] For any vertex $v$ and any level $\ell > \lvl{v}$, it holds $\Phi_v(\ell) \leq \gamma^{\ell} \cdot O(\log^2 n)$ \label{inv:bounded-neighbors-up-to-a-level}
\end{itemize}

Notice that the invariants (a)--(d) are equivalent to the invariants (i)--(iv) of the algorithm from \cite{baswana2011fully} which are adapted to take into account the concept of temporarily free vertices. Invariant (e) formalizes the property of maintaining large support for all edges of the matching. Next we review the four schedulers that are used to maintain the invariants (a)--(f) of the algorithm.

\paragraph*{Scheduler \freeschedule.}
The \freeschedule{} scheduler handles the vertices that become temporarily free due to the adversary. Such vertices reside at $\log_{\gamma}n+1$ queues $\Q_0, \dots, \Q_{\log_{\gamma}n}$ at the different levels of the level-decomposition.
Each subscheduler at level $\ell$ iteratively removes a vertex $v$ from $\Q_\ell$ and calls \handlefree$(v)$, which runs in time $O(\gamma^{\ell})$, simulating $\Delta'$ steps with each update operation.
In \cite{charikar2018fully} the subschedulers at the different levels are executed in an order from the one at the highest level to the one at the lowest level.
In our adaptation to the DMPC model, the $\log_{\gamma} n$ \freeschedule{} subschedulers are executed in parallel. 
Each such subscheduler simulates $\Delta$ operations (in fact, their calls to \handlefree), and the total work by all subschedulers requires a constant number of MPC rounds.
However, this parallel execution of the subschedulers at different levels creates some conflicts that do not appear in \cite{charikar2018fully}, as the subschedulers are assumed to follow some predefined order of execution. We will show how these conflicts can be resolved later on.

\paragraph*{Scheduler \unmatchschedule.}
The goal of the \unmatchschedule{} subscheduler at level $\ell$ is to guarantee that the size of the support of each matched edge at level $\ell$ remains between $\gamma^{\ell}\cdot(1-\epsilon)$ and $\gamma^{\ell}$ (that is, invariant (e) of the algorithm).
Each subscheduler at level $\ell$ simply removes the level-$\ell$ matching edge $e=(u,v)$ of smallest sample space, and executes \handlefree$(u)$ and \handlefree$(v)$.
The computation that is triggered by a removal of a matched edge at level $\ell$ is bounded by $O(\gamma^{\ell})$, and it is executed in batches of $\Delta$ operations.
Each such batch contains exchange of information between $u$ and $v$ and $\Delta$ of their neighbors, and hence, can be executed in $O(1)$ rounds using $\widetilde{O}(1)$ machines and communication per round.

\paragraph*{Scheduler \riseschedule.}
Each subscheduler of this type at level $\ell$ ensures that for each vertex $w$ at level $\ell'<\ell$ it holds that $\Phi_w(\ell) \leq \gamma^{\ell} \cdot O(\log^2 n)$.
The subscheduler, each time identifies the vertex $w$ at level $\ell'<\ell$ with the highest value of $\Phi_w(\ell)$, removes the matching edge $(w,w')$ (if such an edge exists), and raises $w$ to level $\ell$ by executing \riseschedule$(w,\ell)$.
Finally, the subscheduler executes \handlefree$(w)$ and \handlefree$(w')$ to match both $w$ and $w'$. 
The execution of this subscheduler takes $T_\ell = O(\gamma^{\ell})$ time in 
batches of size $\Delta$, that is the subscheduler is executed from $O(\gamma^{\ell} / \Delta)$ update cycles.
Again, each batch of this update can be executed in a constant number of DMPC rounds.

\paragraph*{Scheduler \shuffleschedule.}
This scheduler at level $\ell$ each time picks a matching edge $e=(u,v)$ uniformly at random among the matching edges at level $\ell$, removes it from the matching, and executes \handlefree$(u)$ and \handlefree$(v)$ to try and match again the endpoints of the removed matching edge. The role of this scheduler is mainly to facilitate the proof in \cite{charikar2018fully} showing that the adversary has low probability of deleting a matched edge (the low probability is defined depending on the level of the matched edge). 
While the total time required by this subscheduler is $O(\gamma^{\ell})$, it is executed in batched of size $\Delta' = \Delta\cdot \log n$, which ensures that it runs faster than the \unmatchschedule{} by a logarithmic factor, for technical reasons as explained in \cite{charikar2018fully}.
This scheduler runs only for the levels $\ell$ for which $\gamma^{\ell}/\Delta'>1$ as otherwise each update of at level $\ell$ is executed immediately (not in multiple chunks) and hence the adversary does not have enough time to delete an edge during the execution of a subscheduler at this level.
In a same way as the other schedulers, each batch of this scheduler can be executed in $O(1)$ rounds using $\widetilde{O}(1)$ machines and communication per round.

\smallskip\noindent\textbf{Handling updates.}
Following the insertion of an edge $e=(u,v)$, the algorithm updates the relevant data structures in time $O(\log_{\gamma} n)$.
If both $u$ and $v$ are at level $-1$, the algorithm marks $e$ as matched edge and sets the level of $u$ and $v$ to be $0$ by calling \setlevel$(u,0)$ and \setlevel$(v,0)$. In \cite{charikar2018fully} it is shown that this process is enough to guarantee that all the invariants are maintained, and that an edge insertion can be handled in $O(\log^4n)$ time.
The above process can be simulated in $O(1)$ DMPC rounds, as all instructions involve exchanging information between $u$ and $v$ and their neighbors, as well as each vertex updating their part of the maintained data structures. 

To process the deletion of an edge $e=(u,v)$ we proceed as follows. If the edge does not belong to the matching, it is sufficient to update the relevant data structures which requires only $O(\log n)$ time.
On the other hand, if $e$ belongs to the matching we first remove it from the matching, add its endpoints in $\Q_{\lvl{u}}$.
The above process is sufficient, as the subscheduler \handlefree{} at level $\lvl{u}$ will handle the newly free vertices $u$ and $v$ appropriately.
Moreover, the process ensures that all the invariants maintained by the algorithm continue to be satisfied after this process.

As one can observe, the insertions and deletions of edges do not trigger any long update procedures (even when deleting edges of the matching!), but rather work together with the schedulers in maintaining the invariants (a)--(f) of the algorithm, which in turn ensures that the maintained matching is almost-maximal.
However, as the different subscheduler at the different levels do not communicate with each other but operate independently, there are some issues that arise if they try to process the same vertices.

\subsection{Conflicts between schedulers}
Here we deal with synchronization issues that arise from the fact that all subschedulers are working simultaneously at all times. These issues are called conflicts between subschedulers. We exhibit the conflicts that arise and show the modifications that need to be made in order to ensure that all invariants of the algorithm are satisfied at all times. Some of the modifications were already suggested in \cite{charikar2018fully}, however, we repeat them here for completeness of the algorithm.
In what follows we ignore the overhead added by updating the list $\mathcal{A}$ of active vertices.

\paragraph*{Sampling mates conflict.} The procedure \handlefree~$(v)$ at level $\ell$, as part of a call from its subscheduler, might pick as a new mate of $v$ a vertex that is already processed by some other subscheduler. However, this is not really an issue as the sampling avoids such vertices by sampling from $S(v)\setminus \mathcal{A}$, and the active list $\mathcal{A}$ contains all vertices that are being processed.

\paragraph*{Deleting unmatched edges conflict.}
A conflict may arise when \unmatchschedule{} or \shuffleschedule{} subschedulers try to remove a vertex that has already been removed from the matching. While the case where the processed edge has been removed at a previous round is not really a conflict, as once a subscheduler removes an edge from the matching it informs all other subschedulers in $O(1)$ rounds and using $O(\log n)$ machines, the case where subschedulers from different levels try to remove the same edge from the matching is problematic.
For each \unmatchschedule{} subscheduler we pick the top $2 \log n$ choices of edges to remove from the matching and for each \shuffleschedule{} subscheduler we pick $2\log n$ random edges to remove from the matching.
Next, all \unmatchschedule{} and \shuffleschedule{} subschedulers send their $2 \log n$ choices to the same machine, and that machine decides for which subscheduler removes which edges from the matching, by first processing the \unmatchschedule{} subschedulers in decreasing order of their level followed by the \shuffleschedule{} subschedulers in decreasing order of their level and for each subscheduler we assign the first unassigned choice among its list of $2 \log n$ choices.
Then, the machine communicates to the subscheduler their assigned edges, and hence no conflicts occur among the different subschedulers as each has a unique edge to delete from the matching.

\paragraph*{Match an active vertex conflict.}
A conflict arises if the next vertex chosen by \freeschedule{} subscheduler at level $\ell$ from a queue $\Q_\ell$ is active.
To deal with this issue we delay the scheduling of all the \freeschedule{} subschedulers at least one round (within the same update cycle) after the rest of the subschedulers so that they can send which vertices they mark active in order for them to be removed from the queues $\Q_{\ell}$.

\paragraph{Raising and falling conflicts.}
During subscheduler \riseschedule{}, at level $\ell$, the vertex $v$ that is picked to be raised might be already active.
We do not try to prevent this type of conflicts, as it is possible that we actually want to raise $v$ to level $\ell$ even though $v$ is already active, in order to satisfy the invariant (f) of the algorithm.
In particular, during the process where \riseschedule{} at level $\ell$ chooses a vertex $v$ to move it to level $\ell$, some other procedure might be handling $v$, that is, $v$ might be in the process of being raised or fallen level.
Notice that, if $v$ is being raised or fallen at some level $\ell'>\ell$, then there is no need for \riseschedule{} subscheduler to raise $v$ to $\ell'$.
The case where \riseschedule{} needs to raise $v$ to $\ell$ is when $\ell'<\ell$ (the destination level of $v$ at the process of raising or falling). 

First, we take care of the conflicts between subschedulers of type \riseschedule{}. 
Similarly to the case of the \unmatchschedule{} and \shuffleschedule{} subschedulers, we process the \riseschedule{} subschedulers in a sequence according to their levels and we assign to them (inactive and unassigned) vertices to rise, ensuring that each \riseschedule{} subscheduler at level $\ell$ does not raise the same vertex with one of a \riseschedule{} subschedulers at higher levels.

Other than conflicts between different \riseschedule{} subschedulers, the only other scheduler that might conflict with \riseschedule{} is \handlefree. 
In this case we distinguish conflicts of a \riseschedule{} subscheduler with calls \handlefree$(w)$, where $w$ is being raised, and  calls $\handlefree(w)$, where $w$ is being fallen.
As shown in \cite{charikar2018fully}, the conflicts between \riseschedule{} subscheduler and calls to the procedure \handlefree$(w)$ where $w$ is being raised are avoided as follows. Each level-$\ell$ \riseschedule{} subscheduler picks the subsequent vertex that it is going to raise and adds it into the set of active vertices, so that it cannot be processed by other schedulers. 
The vertex that is going to be raised with the next call to \riseschedule, is called the \emph{next-in-line} vertex of the corresponding subscheduler.
That is, each time a call to \riseschedule{} subscheduler is being initiated, it chooses the vertex that it is going to raise in the next call, and proceeds with raising the vertex that was chosen in the previous call.
It can be shown that this mechanism avoids conflicts between \riseschedule{} and procedure \handlefree, where \handlefree{} is processing a raising vertex.
The correctness is based on the fact that the call to level-$\ell$ \riseschedule{} subscheduler will end later than the call to level-$\ell'$ \handlefree{} procedure, where $\ell'<\ell$.
Moreover, because we schedule the different \riseschedule{} subschedulers in a decreasing order of their level, exactly as it is being done in \cite{charikar2018fully}, our distributed version does not affect their analysis.

Finally, we need to deal with the conflicts that arise between \riseschedule{} subschedulers and calls to procedure \handlefree$(w)$, where $w$ is in the process of falling.
This is a special case on its own, and is not resolved solely by the next-in-line mechanism discussed before, as the call to \handlefree$(w)$ may have been initiated from a level $\ell'>\ell$.
The first modification is that during a call to \handlefree$(w)$ we first check whether $w$ is the next-in-line vertex of any of the \riseschedule{} subschedulers at levels $j > \lvl{w}$, and if yes, we ignore the call to \handlefree$(w)$. 
This trick guarantees that there are no level-$j$ (where $j > \lvl{w}$) \riseschedule{} subscheduler attempts to raise $w$ while $w$ is falling from $\lvl{w}$ to a level $\ell$, as part of a call to \handlefree$(w)$.

It is possible that while $w$ is falling from $\lvl{w}$ to $\ell$, a level-$j$ \riseschedule{} subscheduler attempts to raise $w$ to level $j$.
The next-in-line trick does work here as the call to \handlefree$(w)$ requires more time than \riseschedule{} and hence it is not guaranteed that $w$ will be in a next-in-line for some \riseschedule{} subscheduler with level $\ell<j<\lvl{w}$.
We deal with this by preventing any level-$j$ \riseschedule{} subschedulers to raise $w$ to level $j$ while $w$ is falling for any $j<\lvl{w}$.
Although this guarantees that no \riseschedule{} subscheduler raises a falling vertex, the fact that we prevent the subscheduler to raise a vertex, might violate invariant (f), i.e., that for any vertex $v$ and any level $\ell'> \lvl{v}, \Phi_v(\ell') \leq \gamma^{\ell'} \cdot O(\log^2 n)$.
To ensure that this does not happen, right after $w$ falls to level $\ell$, we immediately raise to the highest level $\ell'$ that violates invariant (f).
It is shown in \cite{charikar2018fully} that this modification prevents the violation of invariant (f) and also the new version of \riseschedule{} subscheduler can be done within the time of the scheduler that initiated the falling of $w$.

\begin{theorem}
A $(2+\epsilon)$-approximate matching can be maintained fully-dynamically in the dynamic MPC model in $O(1)$ rounds per update, using $\widetilde{O}(1)$ active machines per round and $\widetilde{O}(1)$ communication per round, in the worst case. 
\end{theorem}
\begin{proof}
As we argued with the description of each scheduler, we simulate the $\Delta$ or $\Delta'$ operations executed by each subscheduler in \cite{charikar2018fully} with $O(1)$ number of rounds, using $\widetilde{O}(1)$ active machines and $\widetilde{O}(1)$ communication.
Since the job done by the different subschedulers is independent among them, and there are only $O(\log n)$ of these subschedulers, it follows that the execution of all subschedulers in one update cycle can be executed in $O(1)$ rounds, using $\widetilde{O}(1)$ active machines and $\widetilde{O}(1)$ communication.
By the same argument, the authentication process at each update cycle for all subschedulers can be executed in the same time.
Finally, with analogous reasoning, it can be shown that the modifications needed to ensure that no conflicts arise can be executed within the same asymptotic complexity.
\end{proof}
\section{Simulating sequential dynamic algorithms with MPC algorithms}\label{app:red}

\begin{lemma}
	Assume that there exists a sequential algorithm $\mathcal{SA}$ for maintaining a solution to the problem $\mathcal{P}$ with polynomial preprocessing time $p(N)$, and update time $u(N)$, where the algorithm is either deterministic or randomized and the update time is amortized or 
	worst-case. 
	There exists a DMPC algorithm $\mathcal{MRA}$ with $O(p(N))$ number of rounds for the preprocessing, and $O(u(N))$ number of rounds per update with $O(1)$ machines active per round. The DMPC algorithm is of same type as the sequential algorithm. 
\end{lemma}
\begin{proof}
	For the reduction, we assume that the computation is taking place on a single machine $M_{\mathcal{MRA}}$ and the rest of the machines act as the main memory in the corresponding sequential algorithm.
	For each array-based data structure of the sequential algorithm, we allocate a single machine to keep track of how the data are distributed over the machines, i.e., the data structure allocates a minimum number of machines (up to a constant factor) and distributes the data in intervals such that a whole interval of the array lies on a single machine. 
	For each list-based data structure, similarly to the way the sequential algorithm stores a single link to the first element of the list we store only the machine storing the first element together with its position in the memory of the machine. Moreover, at the position of each element of the list we also store a pointer to the machine and position of the next element.
	For other type of data structures we could act similarly.
	For instance if a data structure is a list of a dynamically reallocated array-based data structure, then we could maintain the array-based data structures in as few machines as possible and allocate new machines whenever it is necessary (this is to ensure that we do not use too many machines).
	
	Whenever the algorithm that is executed on $M_{\mathcal{MRA}}$ requests access to an arbitrary position of an array-based data structure, then in a constant number of rounds this memory position is fetched and written back again (if the value has been changed) by only accessing a constant number of machines. In the case where the $\mathcal{MRA}$ algorithm requests access to an element of a list, it is required to specify a pointer to the machine and position of the element (in the beginning a pointer to the first element is specified, and as the list is scanned, the pointer to the next element is known by the algorithm).
	
	The complexity of a sequential algorithm is determined by the number of its accesses to the memory and also by arithmetic operations. Since each access to the memory by $\mathcal{SA}$ is simulated by a constant number of rounds by $\mathcal{MRA}$ with constant number of active machines per round, the running time of $\mathcal{SA}$ is translated to rounds of $\mathcal{MRA}$. 
	Therefore, the preprocessing time $p(N)$ and the update time $u(N)$ of the sequential algorithm can be simulated by $O(p(N))$ and $O(u(N))$ rounds, respectively, by the algorithm $\mathcal{MRA}$ with constant number of machines per round.
\end{proof}

\section{Discussion and Open Problems}
Although we believe that our model is a natural extension of the MPC model for dynamic algorithms, we acknowledge that the DMPC model has a few deficiencies. The main deficiency of the model is that it allows algorithms that during an update make use of a predefined set of machines (in the sense that the same machines are used during this update independently of the nature and content of the update), for instance, the algorithms that make use of a coordinator machine in some of the algorithms presented in this paper.
Such practices might create bottlenecks in the performance of algorithms, and even make the overall system more vulnerable to failures or malicious attacks.
This consideration can be taken into account by adding the following parameter to the DMPC model.
Assuming that the updates are executed uniformly at random from all possible updates at each round, we measure the entropy of the normalized distribution of the total communicated bits among the pairs of machines at each round. The higher the value of the metric the better the algorithm in terms of how uniformly distributed are the transmitted messages among the machines. We next further elaborate on this potential metric.

Consider a particular update round $r$, where the update happening is drawn uniformly at random from the set of all possible update that can potentially happen at round $r$.
Let $\phi:V\times V \rightarrow [C^2]$, where $C$ the total number of machines, be a mapping from pairs of machines to integers, and let $\alpha$ be the vector where $\alpha[\phi(i,j)]$ is the expected size message transmitted from machine $M_i$ to machine $M_j$ at round $r$, which depends on the update happening at round $r$.
For instance, an algorithm using a coordinator machine $M_c$ will have $\sum_{i\not=c}\alpha[\phi(c,i)] = \sqrt{N}$, and hence $M_c$ will be certainly activated and transmitting $\sqrt{N}$ bits in expectation.
Ideally, we would like the expected total communication to be equally distributed over the values $\alpha[\phi(i,j)]$.
This can be captured by the notion of entropy, defined over the normalized vector $\overline{\alpha}$, where $\sum_{i,j}\overline{\alpha}[\phi(i,j)] = 1$.
The entropy $H(\overline{\alpha})$ is maximized when the distribution of the unit value over the entries of $\overline{\alpha}$ is uniform, that is, when $\overline{\alpha}[\phi(i,j)] = 1/\ell$, where $\ell$ the length of $\overline{\alpha}$.
Note that the absolute value of the average value of $\alpha$ is upper bounded by the bound on total communication per round required by our model.
Intuitively, the measure of entropy that we consider quantifies the randomness in the set of machines that exchange messages, with respect to a random update.
For instance, when using a coordinator machine for an algorithm, then the communication is concentrated at the connection between the coordinator and the machine storing the updated elements (which is random), and hence the total entropy is limited. 

A second deficiency of our model is that the measure of total communication per round is somewhat ``coarse'' as it ignores the form of the messages exchanged, e.g., consider a round that uses $O(\sqrt{N})$ total communication and $O(\sqrt{N})$ active machines, in this case, the model does not distinguish between the case where each active machine transmits short messages of size $O(1)$  to another machine and the case where one machine transmits a large message of size  $O(\sqrt{N})$ and the rest of the machines small messages.
Notice that also this second deficiency can be taken into account by introducing the entropy-based metric that we used in the case of the first deficiency.

For the sake of simplicity of the model, we chose to avoid incorporating complicated metrics as parameters of the model in this initial paper introducing the model. However, we believe that this is an interesting direction for future work.arxaa

\paragraph{Open Problems.}
In this paper we initiated the study of algorithms in the DMPC model, by considering some very basic graph problems.
It is natural to study more problems in this setting as the MPC model becomes the standard model for processing massive data sets, and its limitations with respect to processing dynamically generated data are clear.
In general, we think it is of great value to understand the complexity of fundamental problems in the DMPC model, both in terms of upper and lower bounds.
We also believe that it would be interesting to establish connections with other models of computation, in order to develop a better understanding of the strengths and weaknesses of the DMPC model.

\bibliographystyle{abbrv}
\bibliography{bibliography}

\end{document}